\documentclass[journal,draftcls,onecolumn,letterpaper,twoside,12pt]{IEEEtran}
\usepackage{cite,amsmath,amssymb,graphicx}
\usepackage{citesort}
\usepackage{algorithm,algorithmic}
\usepackage{subfigure}
\usepackage{mathrsfs}
\usepackage{url}
\usepackage{bbm}
\usepackage{amsthm}
\usepackage{setspace}
\usepackage{color}
\newtheorem{theorem}{Theorem}

\newtheorem{corollary}{Corollary}

\newtheorem{definition}{Definition}

\newtheorem{proposition}{Proposition}

\newtheorem{lemma}{Lemma}

\newtheorem{remark}{Remark}

\makeatletter
\newcounter{parentalgorithm}
\newenvironment{subalgorithms}{%
  \refstepcounter{algorithm}%
  \protected@edef\theparentalgorithm{\thealgorithm}%
  \setcounter{parentalgorithm}{\value{algorithm}}%
  \setcounter{algorithm}{0}%
  \def\thealgorithm{\theparentalgorithm\alph{algorithm}}%
  \ignorespaces
}{%
  \setcounter{algorithm}{\value{parentalgorithm}}%
  \ignorespacesafterend
}

\newcommand{\ttheta}{\tilde{\theta}}

\newcommand{\Prob}{\mathbb{P}}
\newcommand{\ProbQ}{\mathbb{Q}}

\newcommand{\E}{\mathbb{E}}

\newcommand{\T}{\mathsf{T}}

\newcommand{\ignore}[1]{{}}

\newcommand{{\uh}}{\underline{h}}
\newcommand{{\lh}}{\bar{h}}
\newcommand{\lb}{\left(}
\newcommand{\rb}{\right)}

\begin{document}
\title{Decentralized Sequential Composite Hypothesis Test Based on One-Bit Communication}
\author{Shang Li$^*$, Xiaoou Li$^\dag$, Xiaodong Wang$^*$, Jingchen Liu$^\dag$
\thanks{$^*$S. Li and X. Wang are with Department of Electrical Engineering, Columbia University, New York, NY 10027 (e-mail: \{shang,wangx\}@ee.columbia.edu).
}
\thanks{$^\dag$X. Li and J. Liu are with Department of Statistics, Columbia University, New York, NY 10027 (e-mail: \{xiaoou,jcliu\}@stat.columbia.edu).
}}
\maketitle
\vspace*{-9mm}
\begin{abstract}
This paper considers the sequential composite hypothesis test with multiple sensors. The sensors observe random samples in parallel and communicate with a fusion center, who makes the global decision based on the sensor inputs. On one hand, in the centralized scenario, where local samples are precisely transmitted to the fusion center, the generalized sequential likelihood ratio test (GSPRT) is shown to be asymptotically optimal in terms of the expected stopping time as error rates tend to zero. On the other hand, for systems with limited power and bandwidth resources, decentralized solutions that only send a summary of local samples (we particularly focus on a one-bit communication protocol) to the fusion center is of great importance. To this end, we first consider a decentralized scheme where sensors send their one-bit quantized statistics every fixed period of time to the fusion center. We show that such a uniform sampling and quantization scheme is strictly suboptimal and its suboptimality can be quantified by the KL divergence of the distributions of the quantized statistics under both hypotheses. 
We then propose a  decentralized GSPRT based on level-triggered sampling. That is, each sensor runs its own GSPRT repeatedly and reports its local decision to the fusion center asynchronously. We show that this scheme is asymptotically optimal as the local thresholds and global thresholds grow large at different rates. Lastly, two  particular models and their associated applications are studied to compare the centralized and decentralized approaches. Numerical results are provided to demonstrate that the proposed level-triggered sampling based decentralized scheme aligns closely with the centralized scheme with substantially lower communication overhead, and significantly outperforms the uniform sampling and  quantization based decentralized scheme.

\ignore{
\noindent
To-do list:
\begin{itemize}
\item Introduction and Problem statement 65 pages)
\item Centralized GSPRT (3 pages)
\item Uniform-sampling based GSPRT (4 pages)
\item Level-triggered-sampling based GSPRT (8 pages)
\item Examples and numerical results (8 pages)
\item Conlusion 
\end{itemize}
}
\end{abstract}

\section{Introduction}
It is well known that the sequential hypothesis test generally requires a smaller expected sample size to achieve the same level of error probabilities compared to its fixed-sample-size counterpart. For instance, for testing on different mean values of Gaussian samples, \cite{Poor_DE} showed that the optimal sequential procedure needs four times less samples on average than the Neyman-Pearson test. Following the seminal work \cite{WaldWolf48} that proved the optimality of the sequential probability ratio test (SPRT) in the context of sequential test, a rich body of works has investigated its variants in various scenarios and applications. Among them, the composite hypothesis test is of significant interest. In particular, \cite{Lorden76} generalized SPRT to 2-SPRT; the sequential composite hypothesis test was discussed by \cite{PollakSiegmund75,Lai88,LaiZhang94,Lai01} for the exponential families; furthermore, \cite{Pavlov87,Pavlov90} studied sequential test among multiple composite hypotheses. 

Using multiple sensors for hypothesis test constitutes another mainstream of the sequential inference paradigm, motivated by the potential wide application of wireless sensor technology. In general, multi-sensor signal processing can be divided into two categories. One features parallel structure (also known as the fully distributed scenario), that allows all sensors to communicate based upon a certain network topology and reach consensus by message-passing; the other features a hierarchical structure and requires a fusion center that makes the global decision by receiving information from distributed sensors. In this work, we consider the hierarchical type of systems where sensors play the role of information relay.  In the ideal case, if the system is capable of precisely relaying the local samples from sensors to the fusion center whenever they become available, we are faced with a centralized multi-sensor hypothesis testing problem. However, the centralized setup amounts to instantaneous high-precision communication between sensors and the fusion center (i.e., samples quantized with large number of bits are transmitted at every sampling instant). In practice, many systems, especially wireless sensor networks, cannot afford such a demanding requirement, due to limited sensor batteries and channel bandwidth resources. Aiming at deceasing the communication overhead, many works proposed the decentralized schemes that allow sensors to transmit small number of bits at lower frequency. In particular, \cite{Veeravalli93} described five (``case A'' through ``case E'') scenarios of decentralized sequential test depending on the availability of local sensor memory and feedback from the fusion center to sensors. There, the optimal algorithm was established via dynamic programming for ``case E'' which assumed full local memory and feedback mechanism. However, in resource-constrained sensor networks, it is not desirable for sensors to store large amount of data samples and for the fusion center to send feedback. Therefore, in this paper, we assume that sensors have limited local memory and no feedback information is available. 

As mentioned above, decreasing the communication overhead can be achieved from two perspectives: First, sensors use less bits to represent the local statistics; second, the fusion center samples local statistics at a lower frequency compared to the sampling rate at sensors. On one hand, in many cases, the original sample/statistic is quantized into one-bit message, which is then transmitted to the fusion center. As such, \cite{Tsitsiklis93,Tsitsiklis86} showed that the optimal quantizer for fixed-sample-size test corresponds to the likelihood ratio test (LRT)  on local samples. Then \cite{Nguyen06} demonstrated that the LRT is not necessarily optimal for sequential detection under the Bayesian setting, due to the asymmetry of the Kullback-Leibler   divergence between the null and alternative hypotheses.  \cite{Mei08,Wang11,Wang13} further investigated the stationary quantization schemes under the Bayesian setting. One the other hand, in order to lower the communication frequency, all the above work can be generalized to the case where quantization and transmission are performed every fixed period of time. These schemes generally  involve fixed-sample-size test at sensors and sequential test at the fusion center, which we refer to as the {\it uniform sampling and quantization} strategy.  Alternatively,
\cite{Veeravalli94} proposed that each sensor runs a local sequential test and local decisions are combined at the fusion center in a fixed-sample-size fashion. Furthermore, 
\cite{Hussain94} proposed to run sequential tests at both sensors and the fusion center, amounting to an adaptive transmission triggered by local SPRTs, though no optimality analysis was provided there. To fill that void, 
\cite{Fellouris11} defined such a scheme as level-triggered sampling and proved its asymptotic optimality in both discrete and continuous time. However, 
\cite{Fellouris11,Yasin12,Yasin13,Hussain94} only considered the simple hypothesis test, where the likelihood functions can be specified under both hypotheses. 

In spite of its broad spectrum of applications, the multi-sensor sequential composite hypothesis test remains to be investigated from both algorithmic and theoretical perspectives. Owing to the unknown parameters, the LR-based decentralized algorithms using either uniform sampling or level-triggered sampling as mentioned above are no longer applicable. Hitherto, some existing works have  addressed this problem in the fixed-sample-size setup. For example, \cite{Kar12} developed a binary quantizer by minimizing the worst-case Cramer-Rao bound for multi-sensor estimation of an unknown parameter. Recently, \cite{Fang13} proposed to quantize local samples (sufficient statistics) by comparing them with a prescribed threshold; then, the fusion center performs the generalized likelihood ratio test by treating the binary messages from sensors as random samples. A similar scheme was established in \cite{Ciuonzo13} for a Rao test at the fusion center. Both \cite{Fang13,Ciuonzo13} assumed that the unknown parameter is close to the parameter under the null hypothesis. 
In \cite{Tartakovsky08}, a composite sequential change detection (a variant of sequential testing) based on discretization of parameter space was proposed. 

In this work, we propose two decentralized schemes for sequential composite hypothesis test. The first is a natural extension of the decentralized approach in \cite{Fang13}, that employs the conventional uniform sampling and quantization mechanism, to its sequential counterpart. The second builds on level-triggered sampling and features asynchronous communication between sensors and the fusion center. 
Moreover, our analysis shows that the level-triggered sampling based scheme exhibits asymptotic optimality when the local and global thresholds grow large at different rates, whereas the uniform sampling scheme is strictly suboptimal. Using the asymptotically optimal centralized algorithm as a benchmark\footnote{The performance of the decentralized scheme is supposed to be inferior to that of the centralized one because the fusion center has less information from the local sensors (i.e., a summary of local samples within a period of time, instead of the exact samples at every time instant). 
}, it is found that the proposed level-triggered sampling based scheme yields only slightly larger expected sample size, but with substantially lower communication overhead. 

The key contribution here is that we have applied the level-triggered sampling to the decentralized sequential composite hypothesis test and provided a rigorous analysis on its asymptotic optimality. Though  \cite{SLi14_Rep,SLi14_SG} have applied the level-triggered sampling to deal with multi-sensor/multi-agent sequential change detection problem with unknown parameters, no theoretical optimality analysis was provided there. The main challenge for analysis lies in characterizing the performance of the generalized sequential probability ratio test for generic families of distributions, which has not been fully understood. To that end, the recent work \cite{XLi14} provides the analytic tool that is instrumental to the analysis of the decentralized sequential composite test based on level-triggered sampling in this paper. Note that, in essence, \cite{XLi14} studied the single-sensor sequential composite test, whereas we consider the sequential composite test under the decentralized multi-sensor setup in this paper.

The remainder of the paper is organized as follows. In Section II, we briefly formulate the sequential composite hypothesis test under the multi-sensor setup. Then we discuss the centralized generalized likelihood ratio test in Section III. In Section IV, we propose two decentralized testing schemes based on uniform sampling and level-triggered sampling respectively, together with their performance analysis. Then in Section V, specific models are studied and numerical results are given to further compare the decentralized schemes. Finally, Section VI concludes this paper.

\section{Problem Statement}
Suppose that $L$ sensors observe samples $y_t^\ell, \; \ell=1, \ldots, L,$ at each  discrete time $t$, and communicate to a fusion center which makes the global decision based upon its received messages from sensors, as shown in Fig. \ref{fig:system}. 
\begin{figure}
\centering
\includegraphics[width=0.86\textwidth]{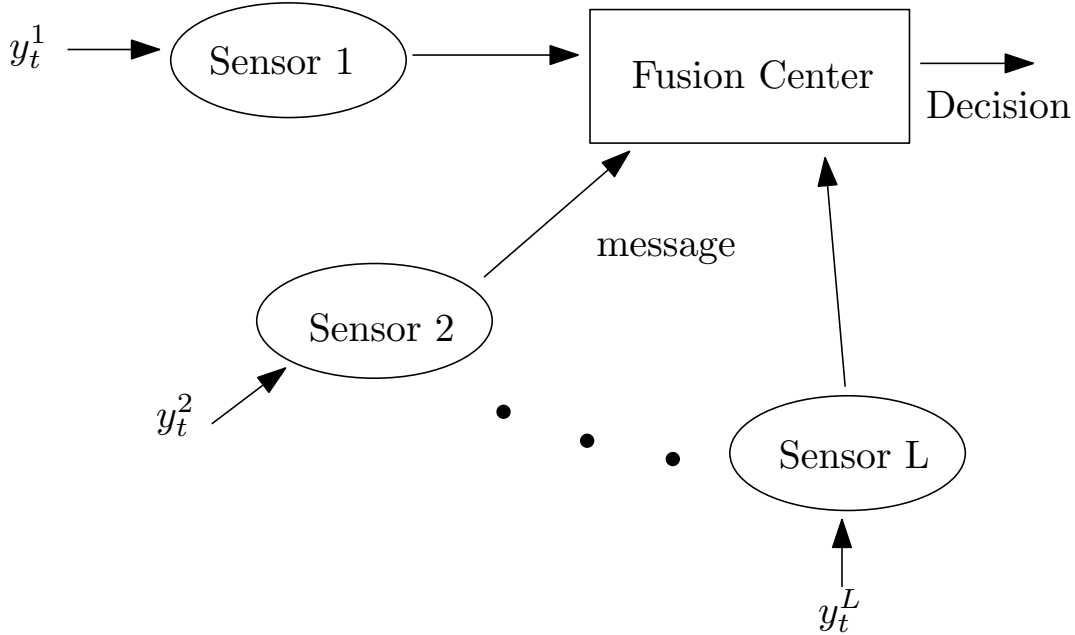}
\caption{A hierarchical multi-sensor system consisting of distributed sensors and a fusion center.}\label{fig:system}
\end{figure}
Assuming the existence of density functions, the observed samples are distributed according to $h_{\gamma}(x)$ under the null hypothesis $\mathcal{H}_0$ and  $f_{\theta}(x)$ under the alternative hypothesis $\mathcal{H}_1$. We assume that $\gamma$ and $\theta$ fall within the parameter sets $\Gamma$ and $\Theta$ respectively. Given $\gamma$ and $\theta$, the random samples under both hypotheses are independent over time and across the sensors. Under such a setup, we arrive at a composite null versus composite alternative hypothesis testing problem:  
\begin{align}\label{simpleVScomposite}
\begin{array}{ll}
\mathcal{H}_0: & y_t^\ell\sim h_{\gamma}\left(x\right), \quad \gamma \in \Gamma, \quad \ell\in {\cal L},\; t=1, 2, \ldots\\
\mathcal{H}_1: & y_t^\ell \sim f_{\theta}\left(x\right), \quad \theta\in \Theta,\quad \ell\in \mathcal{L}, \; t=1, 2, \ldots
\end{array}
\end{align}
where ${\cal L}\triangleq \{1, \ldots, L\}$. In general, $h_\gamma$ and $f_\theta$ may belong to different families of distributions. The goal is to find the stopping time $\T$ that indicates the time to stop taking new samples and the decision function $\delta$ that decides between $\mathcal{H}_0$ and $\mathcal{H}_1$, such that the expected sample size is minimized given the error probabilities are satisfied, i.e.,  
\begin{align}
&\quad\;\inf_\T \qquad {\E_x\T}, \quad x\in \Gamma\cup \Theta\label{samplesize}\\
&\text{subject to}\quad \sup_\gamma \Prob_\gamma\lb\delta=1\rb\le \alpha,\;\; \sup_\theta\Prob_{\theta}\lb\delta=0\rb\le \beta, \label{err_constraints}
\end{align}
where $\E_{\theta}$ denotes expectation taken with respect to (w.r.t.) $f_{\theta}$ and $\E_{\gamma}$ w.r.t. $h_{\gamma}$. Note that \eqref{samplesize}-\eqref{err_constraints} are in fact (possibly uncountably) many optimization problems (depending on the parameter spaces $\Theta$ and $\Gamma$) with the same constraints. Unfortunately, unlike the simple null versus simple alternative hypothesis case, finding a unique optimal sequential test for these problems is infeasible, even when a single-sensor or a centralized setup is considered. Therefore, the approaches that possess asymptotic optimality become the focus of interest. In the following sections, we start by briefly introducing the generalized sequential probability ratio test (GSPRT) as an asymptotically optimal solution for the centralized system; then, two decentralized schemes will be developed based on uniform sampling and level-triggered sampling respectively. In particular, we will show that the latter scheme is asymptotically optimal when certain conditions are met. Here we first give the widely-adopted definition of asymptotic optimality \cite{XLi14,Fellouris11}.
\begin{definition}
Let $\mathcal{T}(\alpha, \beta)$ be the class of sequential tests with
stopping time and decision function $\{\T', \delta'\}$ that satisfy the type-I and type-II
error probability constraints in \eqref{err_constraints}. Then the sequential test $\{\T, \delta\}\in \mathcal{T}(\alpha, \beta)$ is said to be asymptotically optimal, as $\alpha, \beta \to 0$, if
\begin{align}
1\le\frac{\E_x\T}{\inf_{\{\T', \delta'\}\in \mathcal{T}(\alpha, \beta)}\E_x\T'}=1+o_{\alpha, \beta}(1),
\end{align}
or equivalently, $\E_x\T\sim \inf_{\{\T', \delta'\}\in \mathcal{T}(\alpha, \beta)}\E_x\T'$ for every $x\in \Gamma\cup\Theta$. Here, $x\sim y$ denotes ${x}/{y}\to 1$ as $x, y\to \infty$. 
\end{definition}

\ignore{In the following sections, we first consider the scenario where the fusion center collects the exact values of local samples $y_t^\ell$. This assumption means that all sensors transmits local samples with high-precision at every time instant and is referred to as {\it centralized} detection. Then we proceed to investigate a more practical scenario where local sensors only transmit a summary of local samples/statistics, thus  requiring feasible power and communication resources. We refer to the latter scenario as the {\it decentralized} detection. }

\section{Centralized Generalized Sequential Probability Ratio Test}
In this section, we consider the centralized scenario, where local samples $\left\{y_t^\ell\right\}$ are made available at the fusion center in full precision. Note that  the centralized  multi-sensor test is not much different from the single-sensor version except that, at each time instant, multiple samples are observed instead of one. Since finding the optimal sequential composite hypothesis testing is impossible, the solutions with asymptotic optimality become the natural alternatives. In particular, the GSPRT is obtained by substituting the unknown parameter with its maximum likelihood estimate in the SPRT; alternatively, one can perform an SPRT  using the marginal likelihood ratio by integrating out the unknown parameters when the priors on unknown parameters are available. In this paper, we avoid presuming priors on parameters and adopt the GSPRT.

Due to the conditional independence for samples over time and across sensors, the global likelihood ratio function is evaluated as
\begin{align}\label{LLR}
S_t(\gamma, \theta)\triangleq\sum_{\ell=1}^L\sum_{j=1}^ts_j^\ell(\gamma, \theta), \qquad s_j^\ell(\gamma, \theta)\triangleq \log\frac{f_\theta(y_j^\ell)}{{h_\gamma(y_j^\ell)}}\;.
\end{align}
Then the centralized GSPRT can be represented with the following stopping time
\begin{align}
&\T_c\triangleq \inf\left\{t: \widetilde{S}_t\triangleq \log \frac{\max_{\theta\in \Theta}\sum_{\ell=1}^L\sum_{j=1}^tf_\theta(y_j^\ell)}{\max_{\gamma\in \Gamma}\sum_{\ell=1}^L\sum_{j=1}^tf_\gamma(y_j^\ell)}\notin (-B, A)\right\},\label{C_GSPRT_s} 
\end{align}
and the decision function at the stopping instant
\begin{align}
&\delta_{\T_c}\triangleq \left\{
\begin{array}{ll}
1 & \text{if} \quad \widetilde{S}_{\T_c}\ge A,\\
0 & \text{if} \quad \widetilde{S}_{\T_c}\le -B.
\end{array}\right.\label{C_GSPRT_d}
\end{align}
Here $\widetilde{S}_t$ is referred to as the generalized log-likelihood ratio (GLLR) of the samples up to time $t$, and $A, B$ are prescribed constants such that the error probability constraints in  \eqref{err_constraints} are satisfied.  Practitioners can choose their values according to Proposition \ref{C_GSPRT_perf} given below which relates $A, B$ to type-I and type-II error probabilities asymptotically. Before delving into the performance characterization of  the centralized GSPRT \eqref{C_GSPRT_s}-\eqref{C_GSPRT_d},  we recall the Kullback-Leibler (KL) divergence between two distributions $h_{\gamma}$ and $f_{\theta}$:
\begin{align}
D\lb f_{\theta}||h_{\gamma} \rb=\E_{\theta}\lb \log \frac{f_{\theta}\lb Y \rb}{h_{\gamma}(Y)}\rb,\qquad D\lb h_{\gamma}|| f_{\theta} \rb=\E_{\gamma}\lb \log \frac{h_{\gamma}(Y)}{f_{\theta}\lb Y \rb}\rb.
\end{align} 
Assume that the following conditions/assumptions hold,
\begin{itemize}
\item[$A1)$] The distributions under the null and the alternative hypotheses are strictly separated, i.e., $\inf_\gamma D\lb f_\theta || h_\gamma\rb>\varepsilon$ and $\inf_\theta D\lb h_\gamma||f_\theta\rb >\varepsilon$ for some $\varepsilon>0$. This condition implies that the GLLR $\widetilde{S_t}$ takes different drifting directions in expectation under the null and the alternative hypotheses ;
\item[$A2)$] $D\lb f_{\theta} || h_{\gamma}\rb$ and $D\lb h_{\gamma}||f_{\theta}\rb$ are twice continuously  differentiable w.r.t. $\gamma$ and $\theta$;
\item[$A3)$] The parameter spaces $\Gamma$ and $\Theta$ are compact sets;
\item[$A4)$] Let $S(\gamma, \theta)=\log f_{\theta} (Y)-\log h_{\gamma}(Y)$. There exists $\eta>1, x_0$ such that for all $\gamma\in \Gamma, \theta\in \Theta, x>x_0$, we have 
\begin{align}
&\Prob_{\gamma}\lb \sup_{\theta\in\Theta}\left|\nabla_{\theta} S(\gamma, \theta) \right|>x\rb\le e^{-|\log x|^\eta},\label{A3a}\\
\text{and}\qquad&\Prob_{\theta}\lb \sup_{\gamma\in\Gamma}\left|\nabla_{\gamma} S(\gamma, \theta) \right|>x\rb\le e^{-|\log x|^\eta}.\label{A3b}
\end{align}
{This condition imposes that the tail of the first-order derivative of the likelihood ratio w.r.t. $\gamma$ or $\theta$ decays faster than any polynomial.}
\end{itemize}
According to \cite{XLi14}, the performance of the GSPRT can be characterized asymptotically in closed form, which we quote here as a proposition. \begin{proposition}{\cite[Theorem 2.2-2.3]{XLi14}}\label{C_GSPRT_perf}
For the composite hypothesis testing problem given by \eqref{simpleVScomposite}, the GSPRT that consists of stopping rule \eqref{C_GSPRT_s} and decision function \eqref{C_GSPRT_d} yields the following asymptotic performance
\begin{align}
&\sup_{\gamma\in \Gamma}\log \Prob_\gamma(\delta_{\T_c}=1) \sim -A , \qquad
\sup_{\theta\in \Theta}\log \;\Prob_{\theta}\left(\delta_{\T_c}=0\right)\sim -B,\label{GSPRT_ERR}\\
&\E_\gamma\left({\T_c}\right)\sim \frac{B}{\inf_{\theta\in\Theta}D\left(h_\gamma||f_{\theta}\right) L}, \qquad \E_\theta\left({\T_c}\right)\sim\frac{A}{\inf_{\gamma\in \Gamma}D\left(f_{\theta}||h_\gamma\right) L}.\label{GSPRT_ARL}
\end{align}
as $A, B\to \infty$.
\end{proposition}
Proposition \ref{C_GSPRT_perf} indicates that the GSPRT, i.e., \eqref{C_GSPRT_s} and \eqref{C_GSPRT_d}, is asymptotically optimal among the class of $L$-sensor centralized tests $\mathcal{T}_c^L(\alpha, \beta)$ in the sense that 
\begin{align}
\E_x\lb \T_c\rb\sim \inf_{\{\T, \delta\}\in \mathcal{T}_c^L(\alpha, \beta)} \E_x\lb \T\rb, \qquad x\in\Gamma\cup\Theta, 
\end{align}
as $\alpha \triangleq\sup_\gamma \Prob_\gamma\lb \delta_{\T_c}=1\rb\to 0$ and $\beta\triangleq\sup_\theta \Prob_\theta\lb \delta_{\T_c}=0\rb\to 0$ \cite[Corollary 2.1]{XLi14}.
However, as mentioned in Section I, in spite of its asymptotic optimality, the centralized GSPRT yields substantial data transmission overhead between the sensors and the fusion center; therefore, it may  become impractical when the  communication resources are constrained. Moreover, the centralized scheme puts all computation burden at the fusion center. Hence, it is of great interest to consider the decentralized scheme where the computation is distributed among the sensors and the fusion center, with much lower communication overhead between the sensors and the fusion center. 

\ignore{Every sensor transmits a summary of its own samples, in the form of a message taking values in a finite alphabet. Then the fusion center makes a decision on the basis of the messages it receives, i.e., the messages are essentially the observed samples for the fusion center. Given communication constraint, we intend to design a procedure consisting of what message are generated by sensors and when they are transmitted to the fusion center. Such a scheme is of importance in the applications involving geographically distributed sensors. }

\section{Decentralized Sequential Composite Hypothesis Test}
In this section, we investigate the decentralized sequential composite hypothesis test,  where the fusion center is only able to access a summary of local samples. In particular, each sensor transmits a one-bit message to the fusion center every $T_0$ (deterministically or on average) samples. We first consider the conventional decentralized scheme based on the uniform sampling and one-bit quantization. That is, every sensor sends its one-bit quantized local statistic to the fusion center every fixed $T_0$ samples. Then we propose a decentralized scheme based on level-triggered sampling (LTS), where the one-bit transmission is stochastically activated by the local statistic process at each sensor, and occurs every $T_0$ samples {\it on average}. Interestingly, we show that such LTS-based decentralized scheme provably achieves the asymptotic optimality with much lower communication overhead compared with the centralized scheme. 

\subsection{Decentralized GSPRT based on Uniform Sampling and Quantization}
The decentralized scheme based on uniform sampling and quantization is a natural extension of the decentralized fixed-sample-size composite test in \cite{Fang13} to its sequential counterpart. 
Denote the sufficient statistic from the $j$th to the $k$th sample at sensor $\ell$ as $\phi_j^{k, \ell}\triangleq \phi\lb y^\ell_{j}, \ldots, y^\ell_{k}\rb$. On one hand, at every sensor, the statistic is quantized into one-bit message by comparing it with a prescribed threshold $\lambda$, i.e., 
\begin{align}\label{quantization}
q^\ell_n(T_0)\triangleq \text{sign}\lb \phi_{(n-1)T_0+1}^{nT_0, \ell}-\lambda\rb.
\end{align}
Note that \eqref{quantization} corresponds to a stationary quantizer that does not change over time and is studied in decentralized estimation \cite{Ciuonzo13} and detection \cite{Fang13} problems due to its simplicity. On the other hand, the fusion center receives $q_n^\ell, \ell=1, \ldots, L,$ as its own random samples every $T_0$ interval.  To that end, the fusion center runs a GSPRT on the basis of the received $q_n^\ell$'s, which are Bernoulli random variables with different distributions under the null and alternative hypotheses \cite{Blum97}:
\begin{align}
 &\T_q\triangleq \inf \left\{t: \widetilde{G}_t \triangleq \frac{\sup_{\theta\in \Theta} \lb r^t_0\log \lb 1-p_\theta^{T_0}\rb+r_1^t\log p_\theta^{T_0}\rb}{\sup_{\gamma\in \Gamma} \lb r^t_0\log \lb 1-p_\gamma^{T_0}\rb+r_1^t\log p_\gamma^{T_0}\rb} \notin (-B, A)\right\}, \label{q_GSPRT_st}
\end{align}
where $p_x^{T_0}\triangleq \Prob_x\lb q_n^\ell(T_0)=1\rb, \; x\in \{\gamma, \theta\}$, and $r_1^t, r_0^t$ represent the number of received ``$+1$'' and ``$-1$'' respectively, i.e., 
$r^t_0\triangleq \sum_{\ell=1}^L\sum_{n: nT_0\le t}\mathbbm{1}_{\{q_{n}^\ell=1\}}, \; r^t_1\triangleq \sum_{\ell=1}^L\sum_{n: nT_0\le t}\mathbbm{1}_{\{q_{n}^\ell=-1\}}$. Upon stopping, $\mathcal{H}_1$ is declared if $\widetilde{G}_{\T_q}\ge A$, and $\mathcal{H}_0$ is declared if $\widetilde{G}_{\T_q}\le -B$, i.e., $\delta_{\T_q}\triangleq \mathbbm{1}_{\{\widetilde{G}_{\T_q}\ge A\}}$.
Assuming that conditions $A1$-$A4$ listed in the preceding section are satisfied by the Bernoulli random samples $q_n^\ell$, the decentralized GSPRT based on uniform sampling and quantized statistics can be characterized by invoking Proposition \ref{C_GSPRT_perf}.  That is, as $A, B \to \infty
$, the type-I and type-II error probabilities admit
\begin{align}\label{U_error_perf}
 &\sup_{\gamma\in\Gamma}\log \Prob_{\gamma}\left(\delta_{\T_q}=1\right)\sim {-A}, \qquad \sup_{\theta\in \Theta} \log\Prob_\theta\left(\delta_{\T_q}=0\right)\sim {-B},
\end{align}
and the expected sample sizes under the null and alternative hypotheses admit the following asymptotic expressions, respectively:
\begin{align}
&\E_\theta(\T_q)\sim \frac{A}{\lb\inf_\gamma D\lb p^{T_0}_\theta ||p^{T_0}_\gamma\rb /T_0 \rb L}, \label{U_meansize_0}\\
\text{and}\qquad & \E_\gamma(\T_q)\sim \frac{B}{\lb\inf_\theta D\lb p^{T_0}_\gamma||p^{T_0}_\theta\rb/T_0\rb L}.\label{U_meansize_1}
\end{align}
\ignore{\color{blue}
Note that the $A1$-$A4$ condition can be satisfied under some mild condition with the quantizer. In particular, $A1$-$A3$ are natural $A4$ needs some thoughtful attention. For example, a sufficient condition is $p_\gamma^q$ and $p^q_\theta$ are strictly increasing functions of the unknown parameter $\gamma$ and $\theta$ respectively, i.e., $\frac{\partial p^q_\gamma}{\partial \gamma}>0$ and $\frac{\partial p^q_\theta}{\partial \theta}>0$. To see this, we have 
\begin{align}\label{A3a_q}
&\Prob_\gamma\lb \sup_{\theta}\left|\frac{q_n^\ell}{p_\theta^{T_0}}-\frac{1-q_n^\ell}{1-p_\theta^{T_0}}\right|\frac{\partial p_\theta^{T_0}}{\partial \theta}>x\rb\nonumber\\=&\Prob_\gamma\lb \sup_{\theta}\left|\frac{q_n^\ell-p_\theta^{T_0}}{p_\theta^{T_0}\lb1-p_\theta^{T_0}\rb}\right|\frac{\partial p_\theta^{T_0}}{\partial \theta}>x\rb\nonumber\\=&\Prob_\gamma\lb\sup_{\theta}\left|\frac{q_n^\ell-p_\theta^{T_0}}{p_\theta^{T_0}\lb1-p_\theta^{T_0}\rb}\right|\frac{\partial p_\theta^{T_0}}{\partial \theta}>x; q_n^\ell=1\rb+\Prob_\gamma\lb\sup_{\theta}\left|\frac{q_n^\ell-p_\theta^{T_0}}{p_\theta^{T_0}\lb1-p_\theta^{T_0}\rb}\right|\frac{\partial p_\theta^{T_0}}{\partial \theta}>x; q_n^\ell=0\rb\nonumber\\=&\Prob_\gamma\lb \sup_{p_\theta^{T_0}} \frac{1}{p_\theta^{T_0}}>\epsilon x\rb p_\gamma^{T_0}
+\Prob_\gamma\lb \sup_{p_\theta^{T_0}} \frac{1}{1-p_\theta^{T_0}}>\epsilon x\rb \lb1-p_\gamma^{T_0}\rb, \end{align}
where $\epsilon=1/\lb\frac{\partial p_\theta^{T_0}}{\partial \theta}\rb>0$. Note that $0<p_\theta^{T_0}<1$, thus we can always find sufficiently large $x_0$ such that for all $x>x_0$, both terms in \eqref{A3a_q} vanish. Similarly, we can prove the same for \eqref{A3b}. Thus \eqref{A3a}-\eqref{A3b} are satisfied, and Proposition \ref{C_GSPRT_perf} can be applied.
}
It is well known that $D\lb p^{T_0}_\theta ||p^{T_0}_\gamma\rb/T_0<D\lb f_\theta||h_\gamma\rb$ \cite{Tsitsiklis_TC93}, which leads to $\inf_\gamma D\lb p^{T_0}_\theta ||p^{T_0}_\gamma\rb/T_0<\inf_\gamma D\lb f_\theta||h_\gamma\rb$; therefore, the decentralized GSPRT implemented by \eqref{quantization} and \eqref{q_GSPRT_st} yields suboptimal performance, where the suboptimality is determined by the KL divergence between the distributions of quantized sufficient statistics under null and alternative hypotheses. The performance also depends on the choice of the quantization threshold $\lambda$:
\begin{itemize}
\item  $\lambda$ can be chosen such that either $\inf_\gamma D\lb p^{T_0}_\theta||p_\gamma^{T_0}\rb$ or $\inf_\theta D\lb p^{T_0}_\gamma||p^{T_0}_\theta\rb$ is maximized. In general, these two terms cannot be optimized simultaneously. Therefore, a tradeoff is required between the expected sample sizes under the null and alternative hypotheses.
\item Given that typically the expected sample size under the alternative hypothesis is of interest, the optimal $\lambda$, in general, depends on the unknown parameter $\{\theta, \gamma\}$. One possible suboptimal solution is to find the optimal quantizer for the worst-case scenario, i.e., 
\begin{align}\label{opt_la}
\lambda^\star=\arg\;\max_{\lambda} \min_{\theta, \gamma} D\lb p^{T_0}_\theta||p^{T_0}_\gamma\rb.
\end{align}
Nonetheless, the performance is expected to degrade when the actual parameters deviate from the worst-case scenario. 
\end{itemize}
\ignore{
\subsection{GSPRT Based on Local GLRT Decisions}
In this section, we investigate another possible uniform sampling scheme, where every sensor performs generalized likelihood ratio test (GLRT) and send its decision to the fusion center. Similarly as the scheme in last subsection, the fusion center view received one-bit messages as its own random samples, based on which  GSPRT can be performed. Again, let $T_0$ denote the deterministic transmission period and LLR of samples from $j$ to $k$
\begin{align}\label{dLLR}
S_j^k(\theta)\triangleq\sum_{\ell=1}^L\underbrace{\sum_{i=j}^ks_i^\ell(\theta)}_{S_j^{k,\ell}(\theta)}.
\end{align}
Recalling that the GLRT compares the GLR with some prescribed threshold, local one-bit message is defined as
\begin{align}
b^\ell_{n}=\text{sgn}\lb \sup_{\theta}S_{(n-1) T_0+1}^{n T_0, \ell}(\theta)-T_0\lambda\rb, \quad \E_0(s(\theta))<\lambda<\inf_\theta\E_\theta \lb s(\theta)\rb.
\end{align}
Under such a communication protocol, the fusion center compute the LLR of $b_n^\ell$ as the global decision statistic
\begin{align}
J_t(\theta)&=\sum_{n=1}^{N: NT_0\le t}\sum_{\ell=1}^L \log \frac{\Prob_\theta(b_n^\ell)}{\Prob_0(b_n^\ell)}\nonumber\\&=r_0\log\frac{\bar{\beta}_\theta}{1-\bar{\alpha}}+r_1\log\frac{1-\bar{\beta}_\theta}{\bar{\alpha}},
\end{align}
where $r_0$ and $r_1$ again denote the number of received $-1$ and $1$ respectively, and $\bar{\beta_\theta}\triangleq \Prob_\theta(b_n^\ell=0)$ and $\bar{\alpha}\triangleq \Prob_0(b_n^\ell=1)$ are the local error probabilities at each sensor. Since we assume the samples at sensors are conditional i.i.d., their $\bar{\beta}_{\theta}$ and $\bar{\alpha}$ are the same. The corresponding GSPRT on the basis of $b_n^\ell$ are performed to make the global decision:
\begin{align}
 &\T_g\triangleq \inf \left\{t: \widetilde{J}_t\triangleq \sup_{\theta\in \Theta_1} \; J_t(\theta) \notin [-B, A]\right\}, \label{b_GSPRT_st}\\&\delta_g\triangleq \left\{
\begin{array}{ll}
1 & \text{if} \quad \widetilde{J}_{\T_g}(\theta)\ge A,\\
0 & \text{if} \quad \widetilde{J}_{\T_g}(\theta)\le -B.
\end{array}\right.\label{b_GSPRT_d}
\end{align}
We continue to simplify the algorithm by assuming $\bar{\beta}_\theta$ is small for $\theta\in \Theta_1$. Then
\begin{align}
\widetilde{J}_t=r_0 \underline{\Delta}+r_1\bar{\Delta}, 
\end{align}
with quantization numbers defined as
\begin{align}
\underline{\Delta}=\log \frac{1-\bar{\beta}}{\bar{\alpha}}\quad \underline{\Delta}=\log \frac{\bar{\beta}}{1-\bar{\alpha}}, \quad \bar{\beta}\triangleq \sup_\theta\bar{\beta}_\theta
\end{align}
This corresponds to the uniform sampling algorithm discussed in \cite{Yasin12,Yasin13}. Note that $\underline{\Delta}, \bar{\Delta}$ can be computed by offline simulation since they do not admit close-form expressions in general.
\ignore{
\begin{align}
\Prob_0\lb u_n^\ell=1\rb=\Prob_0\lb \sup_{\theta}S_{(n-1) T+1}^{n T, \ell}(\theta)\ge T\lambda\rb\approx \sup_\theta \Prob_0\lb S_{(n-1) T+1}^{n T, \ell}(\theta)\ge T\lambda\rb , \text{\bf *?*}
\end{align}
Then we need to evaluate the right-hand side probability using large deviation theory:
\begin{align}
-\frac{1}{T}\log \Prob_0\lb S_{(n-1) T+1}^{n T, \ell}(\theta)\ge T\lambda\rb \sim \mathcal{I}_0(\theta, \lambda).
\end{align}
\begin{align}
\mathcal{I}_0(\theta, \lambda)\triangleq \sup_{z\ge 0}\left\{\lambda z-\log\E_0\lb e^{zs(\theta)}\rb\right\}, \quad \lambda>\E_0\lb s(\theta)\rb 
\end{align}
By the virtue of Kullback's inequality, we have
\begin{align}
D\lb f_\theta||h_0\rb\ge \mathcal{I}_0\lb \theta, \E_\theta(s(\theta))\rb> \mathcal{I}_0\lb \theta, \lambda \rb 
\end{align}
Then we need to evaluate the right-hand side probability using large deviation theory:
\begin{align}
-\frac{1}{T}\log \Prob_\theta\lb S_{(n-1) T+1}^{n T, \ell}(\theta)\le T\lambda\rb \sim \mathcal{I}_\theta(\theta, \lambda).
\end{align}
\begin{align}
\mathcal{I}_\theta(\theta, \lambda)\triangleq \sup_{z>0}\left\{\lambda z-\log\E_\theta\lb e^{zs(\theta)}\rb\right\}, 
\end{align}
By the virtue of Kullback's inequality, we have
\begin{align}
D\lb h_0||f_\theta\rb\ge \mathcal{I}_\theta\lb \theta, \E_0(s(\theta))\rb> \mathcal{I}_\theta\lb \theta, \lambda \rb 
\end{align}
At the fusion center we have
\begin{align}
{S}_t=\sum_{\ell=1}^L\sum_{n=1}^{N: NT\le t}\log \frac{p_\theta\lb u_{n}^\ell\rb}{p_0\lb u_{n}^\ell\rb}
\end{align}
\begin{align}
\log \frac{p_\theta\lb u_{n}^\ell\rb}{p_0\lb u_{n}^\ell\rb}=\left\{
\begin{array}{ll}
\log \frac{1-\tilde{\beta}_\theta}{\tilde{\alpha}}\to -\log \tilde{\alpha} & \text{if} \;\; u_n^\ell=1\\
\log \frac{\tilde{\beta}_\theta}{1-\tilde{\alpha}} \to \log\tilde{\beta}_{\theta}& \text{if} \;\; u_n^\ell=-1
\end{array}\right.
\end{align}
\begin{align}
\bar{S}_t=\sup_\theta\sum_{\ell=1}^L\sum_{n=1}^{N: NT\le t}\log \frac{p_\theta\lb u_{n}^\ell\rb}{p_0\lb u_{n}^\ell\rb}=\sup_\theta r_0^n \log \bar{\beta}_\theta -r_1^n \log \bar{\alpha}
\end{align}
which boils down to 
\begin{align}
\bar{S}_t= r_0^n \log \sup_\theta \bar{\beta}_\theta -r_1^n \log \bar{\alpha}=r_0^n\log \bar{\beta} - r_1^n\log \bar{\alpha}=\sum_{\ell=1}^L\sum_{n=1}^{N: NT\le t}\lb-\log \bar{\alpha}\rb\mathbbm{1}_{u_n^\ell=1}+\log \bar{\beta} \mathbbm{1}_{u_n^\ell=-1}
\end{align}
\begin{align}
&\E_\theta\lb\bar{s}_t\rb=\lb1-\bar{\beta}_\theta\rb \lb-\log \bar{\alpha}\rb+\bar{\beta}_\theta\log \bar{\beta}\sim -\log \bar{\alpha}\sim T\inf_\theta \mathcal{I}_0\lb\theta, \lambda\rb\\
&\E_0\lb \bar{s}_t\rb=\bar{\alpha} \lb-\log \bar{\alpha}\rb+\lb1-\bar{\alpha}\rb\log \bar{\beta}\sim \log \bar{\beta}\sim -T\inf_\theta \mathcal{I}_\theta\lb \theta, \lambda\rb
\end{align}
\begin{align}
&\E_\theta\lb\T\rb\sim\frac{-\log\alpha}{\inf_\theta \mathcal{I}_0\lb\theta, \lambda\rb}\\
&\E_0\lb\T\rb\sim\frac{-\log\beta}{\inf_\theta \mathcal{I}_\theta\lb\theta, \lambda\rb}
\end{align}
We need to prove $\sup_\lambda\inf_\theta \mathcal{I}_0\lb\theta, \lambda\rb\le D\lb f_\theta||h_0\rb$ and $\sup_\lambda\inf_\theta \mathcal{I}_\theta\lb\theta, \lambda\rb\le D\lb h_0||f_\theta\rb$.}
}

\subsection{Decentralized GSPRT based on Level-Triggering Sampling}
Next, we develop a level-triggered sampling (LTS) scheme for the decentralized  sequential composite test. Here, each sensor runs its own local GSPRT and reports its local decision to the fusion center repeatedly. And a global GSPRT is performed by the fusion center based on the received local decisions from all sensors until a confident decision can be made. As opposed to the uniform sampling scheme, the LTS-based decentralized scheme features asynchronous one-bit communication between local sensors and the fusion center. The idea of running SPRTs at both the sensors and the fusion center was first proposed by \cite{Hussain94} for simple hypothesis test, and was further analyzed in \cite{Fellouris11,Yasin12}. In this work, we apply it to the sequential composite test. The essence of level-triggered sampling is to adaptively update local statistic to the fusion center, i.e., transmit messages only when sufficient information is accumulated, which results in substantially lower communication overhead and superior performance compared with the decentralized scheme based on uniform sampling and finite-bit quantization.  For the simple SPRT, level-triggered sampling is equivalent to Lebesgue sampling of local running LLR. However, since the LLR is not available in the composite case,  we obtain a different procedure than that in \cite{Hussain94,Fellouris11,Yasin12}. Nevertheless, our analysis shows that, in the asymptotic regime, our proposed procedure inherits the same optimality as for the simple test scenario. In the proposed LTS-GSPRT, each sensor employes a sequential procedure instead of a fixed-sample-size procedure. As we show in the following subsections, such a refinement greatly enhances the performance of decentralized detection and leads to the asymptotic optimality.
\subsubsection{LTS-based Approximate GSPRT}
Now we derive the LTS-based decentralized sequential composite testing algorithm. First let us determine the communication protocol and one-bit message at each sensor. Considering that sensors possess limited memory (i.e., scenario A in \cite{Veeravalli93}), every time a local decision is made and transmitted, the corresponding sensor refreshes its memory and runs another GSPRT based on newly arriving samples (Thus the fusion center receives i.i.d. information bits). Then the $n$th transmission time at sensor $\ell$ is a stopping time random variable recursively defined as 
\begin{align}\label{LTS-Sampling}
t^\ell_n \triangleq \inf \left\{t: \widetilde{S}^{t, \ell}_{t^\ell_{n-1}+1} \notin (-{b}, {a}) \right\}, \qquad n=1, 2, \ldots, \; t_0=0,
\end{align}
with 
\begin{align}\label{local_GLR}
\widetilde{S}_{k}^{t,\ell} \triangleq  {\sup_{\theta\in \Theta} \sum_{j=k}^t \log f_\theta(y_j^\ell)}-{\sup_{\gamma \in \Gamma}\sum_{j=k}^t \log h_\gamma(y_j^\ell)},
\end{align}
and ${a}, {b}$ are prefixed constants. Note that \eqref{LTS-Sampling} is equivalent to a local GSPRT at sensor $\ell$, thus different $\{{a}, {b}\}$ lead to different inter-communication period, or sampling frequency by the fusion center.
Correspondingly, the one-bit message amounts to the local decision, i.e., 
\begin{align}\label{LTS-message}
u^\ell_n\triangleq \left\{
\begin{array}{ll}
+1, & \text{if} \; \; \widetilde{S}^{t^\ell_n, \ell}_{t^\ell_{n-1}+1}\ge {a}\; ,\\
-1, & \text{if} \; \; \widetilde{S}^{t^\ell_n, \ell}_{t^\ell_{n-1}+1} \le -{b}\; .\\
\end{array}\right.
\end{align}
Intuitively, \eqref{LTS-Sampling}-\eqref{LTS-message} indicate that sensors run GSPRT repeatedly in parallel and their decisions are transmitted to the fusion center in an asynchronous fashion. Given the level-triggered sampling scheme at sensors, we proceed to define an approximation to the GLLR at the fusion center,
\begin{align}
\widetilde{V}_t=\sum_{\ell=1}^L \sum_{n=1}^{N^\ell_t}\lb{a}\mathbbm{1}_{\{u^\ell_{n}=1\}}-{b}\mathbbm{1}_{\{u^\ell_{n}=-1\}}\rb\label{eq:bS},
\end{align}
where $N^{\ell}_t=\max\{n: t_n^{\ell}\leq t\}$.
The fusion center stops receiving messages at the stopping time
\begin{align}\label{eq:D-GSPRT}
\T_p\triangleq \inf\left\{t: \widetilde V_t\notin (-B, A)\right\},
\end{align}
and makes the decision
\begin{align}
&\delta_{\T_p}\triangleq \left\{
\begin{array}{ll}
1 & \text{if} \quad \widetilde V_{\T_p}\geq A,\\
0 & \text{if} \quad \widetilde V_{\T_p}\leq -B.
\end{array}\right.\label{D_GSPRT_d}
\end{align}
In effect, as we will see later, \eqref{eq:D-GSPRT} amounts to an approximation to the GSPRT at the fusion center based on the received one-bit messages $\left\{u_n^\ell\right\}$.
The proposed decentralized sequential composite test procedure based on level-triggered sampling  is summarized as Algorithm \ref{lts_alg}-\ref{center_alg}.
\begin{subalgorithms}
\begin{algorithm}
\caption{\bf : Repeated GSPRT at Local Sensors}
\begin{algorithmic}[1]\label{lts_alg}
\STATE Initialization: $t\leftarrow 0, t_s\leftarrow 1, \widetilde{S}^{\ell}\leftarrow 0$
\STATE {\bf while} $\widetilde{S}^{\ell}\in (-{{b}}, {{a}})$ {\bf do}
\STATE \quad $t \leftarrow t+1$ and take new sample $y^\ell_t$
\STATE \quad Compute $\widetilde{S}^\ell= \widetilde{S}_{t_s}^{t, \ell}$ according to \eqref{local_GLR}
\STATE {\bf end while}
\STATE $t_s\leftarrow t$
\STATE Send $u^{\ell}=\mathbbm{1}_{\{\widetilde{S}^\ell\ge {a}\}}-\mathbbm{1}_{\{\widetilde{S}^\ell\le -{b}\}}$ to the fusion center
\STATE Reset $\widetilde{S}^{\ell}\leftarrow 0$ and go to line 2.
\end{algorithmic}
\end{algorithm}
\begin{algorithm}
\caption{\bf : Global GSPRT at Fusion Center}
\begin{algorithmic}[1]\label{center_alg}
\STATE Initialization: $\widetilde{V}\leftarrow 0$
\STATE {\bf while} $-B<\widetilde{V}< A$ {\bf do}
\STATE \quad Listen to the sensors and receive information bits, say, $r_0$ ``$+1$''s and $r_1$ ``$-1$''s
\STATE \quad $\widetilde{V}\leftarrow\widetilde{V}+r_1{a}-r_0{b}$
\STATE {\bf end while}
\STATE  {\bf if} $\widetilde{V}\ge A$ {\bf then} decide $\mathcal{H}_1$ \\
\STATE  {\bf else} decide $\mathcal{H}_0$
\end{algorithmic}
\end{algorithm}
\end{subalgorithms}
\subsubsection{A Closer Look at the LTS-based Approximate GSPRT} Next we discuss how Algorithm~\ref{lts_alg}-\ref{center_alg} approximates the optimal procedure, i.e., GSPRT, at the fusion center. 
The optimal rule at the fusion center is to compute the LLR of  the local GSPRT decisions, i.e., 
\begin{align}
&V_t(\gamma, \theta)=\sum_{\ell=1}^L \sum_{n=1}^{N^\ell_t}v^\ell_{n}(\gamma, \theta)\label{LTS_Bernoulli}\\
\text{and}\quad
&v^\ell_{n}(\gamma, \theta)\triangleq \left\{
\begin{array}{ll}
\log \frac{1-\widetilde{\beta}_\theta}{\widetilde{\alpha}_\gamma}& \text{if}\;\; u_{n}^\ell=1\;, \\
\log \frac{\widetilde{\beta}_\theta}{1-\widetilde{\alpha}_\gamma}& \text{if}\;\; u_{n}^\ell=-1\;,
\end{array}\right.\label{LTS_Bernoulli_LLR}
\end{align}
where $v^\ell_n(\gamma, \theta)$ is the LLR of the Bernoulli sample $y_n^\ell$, and $\widetilde{\alpha}_\gamma$ and $\widetilde{\beta}_\theta$ are the type-I and type-II error probabilities respectively at the local sensor, i.e., 
\begin{align}\label{eq:local_error}
\widetilde{\alpha}_\gamma\triangleq \Prob_\gamma(u_{n}^\ell=1), \quad \widetilde{\beta}_\theta\triangleq \Prob_\theta(u_{n}^\ell=-1).
\end{align}
Note that ${V}_t(\gamma, \theta)$ is again a function of the unknown parameters since the distribution of $u^\ell_n$ varies with $\gamma$ and $\theta$. To that end, employing the GSPRT as that in \eqref{C_GSPRT_s} and \eqref{q_GSPRT_st}, the original global stopping time is expressed as
\begin{align}\label{LTS_exact_GSPRT}
\inf\left\{t: \inf_{\gamma} \sup_{\theta} V_t(\gamma, \theta)\notin (-B, A)\right\}.
\end{align}
The global GSPRT involves solving the maximization in \eqref{LTS_exact_GSPRT} whenever a new message $u_{n}^\ell$ is received. However, unlike in the uniform sampling case, solving this optimization problem is no easy task since the distribution of $u^\ell_n$ as a function of $\theta$ and $\gamma$ is unclear. Aiming for a computationally feasible algorithm, we continue to simplify \eqref{LTS_exact_GSPRT} in what follows. Using \eqref{LTS_Bernoulli}-\eqref{LTS_Bernoulli_LLR}, 
\begin{align}
 \inf_\gamma\sup_\theta V_t(\gamma, \theta)\nonumber
 &=\inf_\gamma\sup_\theta\sum_{\ell=1}^L \sum_{n=1}^{N^\ell_t}\lb{\log \frac{1-\widetilde{\beta}_\theta}{\widetilde{\alpha}_\gamma}\mathbbm{1}_{\{u^\ell_{n}=1\}}+\log \frac{\widetilde{\beta}_\theta}{1-\widetilde{\alpha}_\gamma}\mathbbm{1}_{\{u^\ell_{n}=-1\}}}\rb\nonumber\\
 &\sim \inf_\gamma\sup_\theta\sum_{\ell=1}^L \sum_{n=1}^{N^\ell_t}\lb{-\log {\widetilde{\alpha}_\gamma}\mathbbm{1}_{\{u^\ell_{n}=1\}}+\log {\widetilde{\beta}_\theta}\mathbbm{1}_{\{u^\ell_{n}=-1\}}}\rb\nonumber\\&=\sum_{\ell=1}^L \sum_{n=1}^{N^\ell_t}\lb{-\sup_{\gamma\in \Gamma}\log {\widetilde{\alpha}_\gamma}\;\mathbbm{1}_{\{u^\ell_{n}=1\}}+\sup_{\theta\in \Theta}\log {\widetilde{\beta}_\theta}\;\mathbbm{1}_{\{u^\ell_{n}=-1\}}}\rb, \;\;\text{as} \;\; {a}, {b}\to\infty,\label{D-approxGLLR}
\end{align}
Therefore, denoting $\widetilde{\alpha}\triangleq \sup_\gamma \widetilde{\alpha}_\gamma, \widetilde{\beta}\triangleq \sup_\theta \widetilde{\beta}_\theta$, the global GLLR is approximately a simple random walk process
\begin{align}
\sum_{\ell=1}^L \sum_{n=1}^{N^\ell_t}\lb{-\log {\widetilde{\alpha}}\;\mathbbm{1}_{\{u^\ell_{n}=1\}}+\log {\widetilde{\beta}}\;\mathbbm{1}_{\{u^\ell_{n}=-1\}}}\rb
\sim\sum_{\ell=1}^L \sum_{n=1}^{N^\ell_t}\lb{a}\mathbbm{1}_{\{u^\ell_{n}=1\}}-{b}\mathbbm{1}_{\{u^\ell_{n}=-1\}}\rb ,
\end{align}
due to Proposition~\ref{C_GSPRT_perf}. The above expression implies that the stochastic process $\widetilde{V}_t$ as defined in \eqref{eq:bS}
approximates the GLLR $\inf_\gamma\sup_\theta V_t(\gamma, \theta)$.

\subsection{Performance Analysis of the LTS-based Decentralized Scheme}
In this subsection, we show that the LTS-based decentralized scheme serves as a superior solution to the uniform-sampling-based scheme because it preserves the asymptotic optimality of the centralized scheme. This interesting property allows us to achieve the same centralized asymptotic performance, but consuming significantly lower communication resources. In particular, the expected sample size under the null and  alternative hypotheses are characterized asymptotically by the following theorem.
\begin{theorem}
In the asymptotic regime where ${b}, {a}\to \infty$ and $A/{a}, B/{b} \to \infty$, the expected sample sizes of LTS-GSPRT admit the following asymptotic expressions
\begin{align}
&\E_\gamma \lb\T_p\rb  \sim \frac{B}{\inf_\theta{D}\lb h_\gamma||f_\theta\rb L}+o\lb B\rb,\label{mean_size_0}\\
&\E_\theta \lb\T_p\rb \sim \frac{A}{\inf_\gamma{D}\lb f_\theta||h_\gamma\rb L}+o\lb A\rb.\label{mean_size_1}
\end{align}
\end{theorem}
\begin{proof}
See Appendix A.
\end{proof}
Notably, as opposed to that of the uniform sampling scheme in \eqref{U_meansize_0}-\eqref{U_meansize_1}, the expected sample sizes of the proposed LTS-based decentralized scheme preserve the KL divergences between $f_\theta$ and $h_\gamma$ as the denominators. In fact, $\E_\gamma\lb\T_p\rb$ and $\E_\theta\lb\T_p\rb$ increase with $A$ and $B$ at the same rate as that of the centralized GSPRT (cf. \eqref{GSPRT_ARL}). We next proceed to relate the type-I and type-II error probabilities of the LTS-based decentralized scheme to the global decision thresholds $\{-B, A\}$ by the theorem below.
\begin{theorem}
In the asymptotic regime where ${b}, {a}\to \infty, \lim\sup {a}/{b}<\infty$, $\lim\sup {b}/{a}<\infty$ and $A/{a}, B/{b} \to \infty$, the type-I and type-II error probabilities of the LTS-GSPRT admit the following asymptotic expressions:
\begin{align}
&\sup_{\gamma\in \Gamma}\log \Prob_\gamma(\delta_{\T_p}=1) \sim -A , \label{type-I}\\
&\sup_{\theta\in \Theta}\log \;\Prob_{\theta}\left(\delta_{\T_p}=0\right)\sim -B.\label{type-II}
\end{align}
\end{theorem}
\begin{proof}
See Appendix B.
\end{proof}
Combining \eqref{mean_size_0}-\eqref{type-II}, we arrive to the following conclusion on the asymptotic optimality of the proposed LTS-based decentralized algorithm.
\begin{corollary}
Let $\mathcal{T}_d^L\lb\alpha, \beta\rb$ be the class of any $L$-sensor decentralized sequential tests, of which the type-I and type-II error probabilities are bounded by $\alpha$ and $\beta$ respectively. Then the proposed LTS-based GSPRT $\{\T_p, \delta_{\T_p}\}$ is asymptotically optimal within this class, i.e., 
\begin{align}
\E_x\lb\T_p\rb\sim \inf_{\{\T, \delta\}\in \mathcal{T}_d^L\lb\alpha, \beta\rb} \E_x\lb\T \rb, \qquad x\in \Gamma\cup\Theta,
\end{align}
as $\alpha\triangleq \sup_\gamma\Prob_\gamma\lb\delta_{\T_p}=1\rb\to 0$ and $\beta\triangleq \sup_\theta\Prob_\theta\lb\delta_{\T_p}=0\rb\to 0$.
\end{corollary}
\begin{proof}
Given the same error probabilities $\alpha=\sup_\gamma\Prob_\gamma\lb \delta_{\T_c}=1\rb=\sup_\gamma\Prob_\gamma\lb \delta_{\T_p}=1\rb$ and  $\beta=\sup_\theta\Prob_\theta\lb \delta_{\T_c}=0\rb=\sup_\theta\Prob_\theta\lb \delta_{\T_p}=0\rb$, the expected sample sizes of the centralized and LTS-based decentralized scheme $\T_p$ admit the following asymptotic performance, as $\alpha, \beta \to \infty$:
\begin{align}
&\E_\gamma\T_c\sim \E_\gamma\T_p\sim\frac{-\log \beta}{\inf_\theta D\lb h_\gamma||f_\theta\rb L},\label{asym_c}\\&\E_\theta\T_c\sim \E_\theta\T_p\sim \frac{-\log \alpha}{\inf_\gamma D\lb f_\theta ||h_\gamma\rb L}. \label{asym_d}
\end{align}
These expressions suggest that the LTS-based decentralized scheme inherits the asymptotic performance of the centralized GSPRT. As a result, it is also safe to say that LTS-GSPRT is asymptotically optimal among the decentralized schemes that satisfy the same error rate constraints, since
\begin{align}
1\le \frac{\E_x\T_p}{\inf_{\{\T, \delta\}\in \mathcal{T}_L^d(\alpha, \beta)}\E_x\T}&\le \frac{\E_x\T_p}{\inf_{\{\T, \delta\}\in \mathcal{T}_L^c(\alpha, \beta)}\E_x\T}\nonumber\\&=\frac{\E_x\T_p}{\E_x\T_c}\frac{\E_x\T_c}{\inf_{\{\T, \delta\}\in \mathcal{T}_L^c(\alpha, \beta)}\E_x\T}\sim 1+o_{\alpha, \beta}(1).
\end{align}
The second inequality holds true by noting that no decentralized scheme can outperform the centralized one because less information is available at the fusion center. The last asymptotic relation is obtained by using \eqref{asym_c}-\eqref{asym_d} and the conclusion in Proposition 1, i.e., $\E_x\T_c\sim \inf_{\{\T, \delta\}\in \mathcal{T}_L^c(\alpha, \beta)}\E_x\T$.
\end{proof}

Recall that, for the simple null versus simple alternative hypothesis test, where the SPRT is optimal, the centralized and LTS-based decentralized SPRT (denoted as $\tau_c$ and $\tau_p$ respectively) provide the following asymptotic performance \cite{Yasin12,Yasin13}:
\begin{align}
&\E_\gamma\tau_c\sim \E_\gamma\tau_p\sim\frac{-\log \beta}{D\lb h_\gamma||f_\theta\rb L},\label{s_asym_c}\\&\E_\theta\tau_c\sim \E_\theta\tau_p\sim \frac{-\log \alpha}{D\lb f_\theta ||h_\gamma\rb L}, \label{s_asym_d}
\end{align} 
where $\alpha\triangleq \Prob_\gamma\lb\delta=1\rb, \; \beta\triangleq \Prob_\theta\lb\delta=0\rb$. Compared to the simple test where parameter values are given, the proposed sequential composite test requires larger expected sample sizes under both hypotheses (since the expected sample sizes are inversely proportional to $\inf_{\gamma\in \Gamma}D\lb f_\theta||h_\gamma\rb$, $\inf_{\theta\in \Theta}D\lb h_\gamma||f_\theta\rb$ instead of $D\lb f_\theta||h_\gamma\rb$, $D\lb h_\gamma||f_\theta\rb$, as seen by comparing \eqref{asym_c}-\eqref{asym_d} and \eqref{s_asym_c}-\eqref{s_asym_d}). This is the price we pay for not knowing the exact parameters.


\section{Numerical Results}
There are a wide range of applications where the decentralized sequential composite hypothesis test plays an important role. In this section, we apply the proposed centralized and decentralized sequential tests to two examples: one is to detect the mean shift of Gaussian random samples; and the other involves spectrum sensing in cognitive radio systems. 

\subsection{Mean Shift of Gaussian Random Samples}
Detecting the mean shift of Gaussian random samples has many applications. For example, suppose we intend to detect the presence of a unknown parameter $\theta$ as soon as possible in the environment contaminated by white Gaussian noise. Here $\theta$ could be the energy of an object that is monitored by a wireless sensor network or a multi-station radar system. The target parameter is assumed to be within a certain interval, i.e., $\theta\in [\theta_0, \theta_1], \theta_0>0$. Then we have an $L$-sensor hypothesis testing problem:
\begin{align}
\begin{array}{ll}
\mathcal{H}_0: & y_t^\ell=e_t^\ell, \quad \ell\in {\cal L},\; t=1, 2, \ldots\\
\mathcal{H}_1: & y_t^\ell=\theta+e_t^\ell,\quad 0<\theta_0\le \theta\le \theta_1,\quad \ell\in {\cal L}, \; t=1, 2, \ldots
\end{array}
\end{align}
where $e_t^\ell\sim\mathcal{N}(0,\sigma^2)$. Sensors are able to transmit one-bit every $T_0$ sampling instants on average. For this model, both $f_\theta$ and $h_\gamma$ are Gaussian probability density functions and $\gamma=0$. The sufficient statistic of the $j$th to $k$th samples at sensor $\ell$ is their summation, denoted as $\phi_j^{k, \ell}=\mathcal{S}_{j}^{k, \ell}\triangleq\sum_{i=j}^k y_i^{\ell}$. First of all, we verify that the log likelihood ratio of $y_t^\ell$, i.e., 
\begin{align}
S(\gamma, \theta)=\left({\left(\theta-\gamma\right)}y_t^\ell-\frac{\theta^2}{2}+\frac{\gamma^2}{2}\right)/\sigma^2
\end{align}
satisfies the conditions $A1$-$A4$. While conditions $A2$-$A3$ are easily verified, conditions $A1$ and $A4$ require the following check: 
\begin{itemize}
\item The KL divergence admits
$D\lb f_\theta||h_\gamma\rb=D\lb h_\gamma || f_\theta\rb={\lb\theta-\gamma\rb^2}/{(2\sigma^2)}$.
By choosing $0<\varepsilon<\frac{\theta_0^2}{2\sigma^2}$, we have $D(f_\theta||h_0)=\frac{\theta^2}{2\sigma^2}>\varepsilon$ and $\inf_{\theta_0\le\theta\le\theta_1} D(h_0||f_\theta)=\frac{\theta_0^2}{2\sigma^2}>\varepsilon$;
\item For \eqref{A3a}, let $x>x_0\ge\frac{\theta_1-\theta_0}{2\sigma^2}$, then we have
\begin{align}
&\Prob_\gamma\lb \sup_{\theta_0\le\theta\le\theta_1} |\nabla_{\theta}S(\theta, \gamma)| >x\rb \nonumber\\=& \Prob_\gamma\lb \sup_{\theta_0\le\theta\le\theta_1} |y_t^\ell-\theta|>x\sigma^2 \rb \nonumber\\=&\Prob_\gamma\lb |y_t^\ell-\theta_0|>x\sigma^2; \; y_t^\ell\ge \frac{\theta_0+\theta_1}{2}\rb+\Prob_\gamma\lb |y_t^\ell-\theta_1|>x\sigma^2;\; y_t^\ell< \frac{\theta_0+\theta_1}{2}\rb\nonumber\\=&\Prob_\gamma\lb y_t^\ell>x\sigma^2+\theta_0\rb+\Prob_\gamma\lb y_t^\ell<-x\sigma^2+\theta_1\rb\nonumber\\=&\Phi\lb-\frac{x\sigma^2+\theta_0-\gamma}{\sigma}\rb+\Phi\lb\frac{-x\sigma^2+\theta_1-\gamma}{\sigma}\rb
\end{align}
Note that $\Phi\lb {-x}\rb\sim e^{-x^2}$ for large $x$, hence we can always find a sufficiently large $x_0\ge \frac{\theta_1-\theta_0}{2\sigma^2}$ such that $x^2>|\log x|^\eta$ , or equivalently, $\Prob_\gamma\lb \sup_{\theta_0\le\theta\le\theta_1} |\nabla_{\theta}S(\theta, \gamma)| >x \rb\le e^{-|\log x|^\eta}$ for $x>x_0, \eta>1$. Similarly, we can show that \eqref{A3b} holds as well.
\end{itemize}
Therefore, Proposition 1 and Theorems 1-2 can be applied to characterize the asymptotic performance of the centralized GSPRT and LTS-based GSPRT for the  problem under consideration.

To implement the centralized GSPRT in \eqref{C_GSPRT_s}-\eqref{C_GSPRT_d}, the global GLLR at the fusion center is computed as
\begin{align}\label{example1_GLLR}
\widetilde{S}_{j}^{k}&=\sup_{\theta\ge \theta_0} \lb\theta \sum_{\ell=1}^L\mathcal{S}^{k, \ell}_{j}-L\lb k-j+1\rb\frac{\theta^2}{2}\rb/\sigma^2\nonumber\\&=\lb\hat{\theta}_j^k\sum_{\ell=1}^L\mathcal{S}_j^{k, \ell}-L\lb k-j+1\rb\frac{\lb\hat{\theta}_j^k\rb^2}{2}\rb/\sigma^2,
\end{align}
with $\hat{\theta}^k_j=\mathcal{E}\lb \sum_{\ell=1}^L\mathcal{S}_{j}^{k, \ell}/\lb k-j+1\rb/L, \theta_0, \theta_1\rb$, and  
\begin{align}
\mathcal{E}(x, \theta_0, \theta_1)\triangleq \left\{
\begin{array}{ll}
x, & \text{if}\; x\in [\theta_0, \theta_1],\\
\theta_1, & \text{if}\; x>\theta_1,\\
\theta_0, & \text{if}\; x<\theta_0.
\end{array}\right.
\end{align}
Substituting $\widetilde{S}_t$ in \eqref{C_GSPRT_s}-\eqref{C_GSPRT_d} with $\widetilde{S}_1^t$ computed by \eqref{example1_GLLR}  gives the centralized GSPRT (C-GSPRT). 

For the LTS-based GSPRT (LTS-GSPRT), note that the parameter MLE at sensor $\ell$ based on the $j$th to $k$th samples is straightforwardly computed as 
$\hat{\theta}_{j}^{k, \ell}=\mathcal{E}\lb\mathcal{S}_{j}^{k, \ell}/\lb k-j+1\rb, \theta_0, \theta_1\rb$,
which leads to the local GLLR statistic at sensor $\ell$:
\begin{align}\label{example_LLR}
\widetilde{S}^{t, \ell}_{j}=\lb\hat{\theta}_{j}^k\mathcal{S}_{j}^{k, \ell}-(k-j+1)\frac{\lb \hat{\theta}_{j}^{k, \ell}\rb^2}{2}\rb/\sigma^2.
\end{align}
Substituting \eqref{example_LLR} into \eqref{LTS-Sampling} and \eqref{local_GLR}, the LTS-GSPRT can be implemented according to Algorithm 1a-1b. 
\ignore{In addition, the KL divergences that determine the performance of both C-GSPRT and LTS-GSPRT are given by
\begin{align}
&D\lb f_\theta||h_0\rb=\frac{\theta^2}{2\sigma^2}\,, 
\quad \inf_{\theta\ge \theta_0}D\lb h_0||f_\theta\rb=\frac{\theta_0^2}{2\sigma^2}\,.
\end{align}}

To implement the uniform sampling based GSPRT (U-GSPRT), we quantize the sufficient statistics $\mathcal{S}_{(n-1)T_0+1}^{nT_0,\ell}$ at the $n$th transmission period at local sensors by
\begin{align}\label{example_quantizer}
q_n^\ell=\text{sign}\left(\mathcal{S}_{(n-1)T_0+1}^{nT_0,\ell}-\lambda\right). 
\end{align}
Given the threshold $\lambda$, and the distribution of statistic 
\begin{align}
\mathcal{S}_{(n-1)T_0+1}^{nT_0,\ell}\sim\left\{
\begin{array}{ll}
\mathcal{N}\lb 0, \sigma^2T_0 \rb&\text{under}\;\;\mathcal{H}_0,\\
\mathcal{N}\lb \theta T_0, \sigma^2T_0 \rb&\text{under}\;\;\mathcal{H}_1,
\end{array}\right.
\end{align}
we have the distribution of Bernoulli samples as
\begin{align}
\Prob_x\lb q_n^\ell=1\rb=p^{T_0}_x(\lambda)=1-\Phi\lb\frac{\lambda-xT_0}{\sigma\sqrt{T_0}}\rb, \quad x\in \{0, [\theta_0, \theta_1]\}.
\end{align}
Again we first verify that the log likelihood ratio of $q_n^\ell$, i.e.,
\begin{align}\label{U_LLR_1}
S_u(\theta, \gamma)=q_n^\ell \log \frac{p_\theta^{T_0}(\lambda)}{p_\gamma^{T_0}(\lambda)}+\lb1-q_n^\ell\rb \log \frac{1-p_\theta^{T_0}(\lambda)}{1-p_\gamma^{T_0}(\lambda)}
\end{align}
satisfies conditions $A1$-$A4$. Specifically, $A2$-$A3$ is easy to verify, and we check $A1$ and $A4$ as follows: 
\begin{itemize}
\item {Since $p_\theta^{T_0}\neq p_\gamma^{T_0}$ for all $\theta_0\le \theta\le \theta_1$ and $\gamma=0$, it is guaranteed that $D\left(p_\theta^{T_0}||p_0^{T_0}\right)$ and thus $\inf_\theta D\left(p_0^{T_0}||p_\theta^{T_0}\right)$ are positive, and there exists an $\varepsilon>0$ such that $D\left(p_\theta^{T_0}||p_0^{T_0}\right)>\varepsilon$ and $\inf_\theta D\left(p_0^{T_0}||p_\theta^{T_0}\right)>\varepsilon$.}
\item To verify \eqref{A3a} for $S_u\lb\gamma, \theta\rb$, we have
\begin{align}\label{A3a_q}
&\Prob_\gamma\lb \sup_{\theta_0\le\theta\le\theta_1} |\nabla_{\theta}S_u(\theta, \gamma)| >x\rb \nonumber\\=&\Prob_\gamma\lb \sup_{\theta_0\le\theta\le\theta_1}\left|\frac{q_n^\ell}{p_\theta^{T_0}}-\frac{1-q_n^\ell}{1-p_\theta^{T_0}}\right|\frac{\partial p_\theta^{T_0}}{\partial \theta}>x\rb\nonumber\\=&\Prob_\gamma\lb \sup_{\theta_0\le\theta\le\theta_1}\left|\frac{q_n^\ell-p_\theta^{T_0}}{p_\theta^{T_0}\lb1-p_\theta^{T_0}\rb}\right|\frac{\partial p_\theta^{T_0}}{\partial \theta}>x\rb\nonumber\\=&\Prob_\gamma\lb\sup_{\theta_0\le\theta\le\theta_1}\left|\frac{q_n^\ell-p_\theta^{T_0}}{p_\theta^{T_0}\lb1-p_\theta^{T_0}\rb}\right|\frac{\partial p_\theta^{T_0}}{\partial \theta}>x; q_n^\ell=1\rb\nonumber\\&\quad+\Prob_\gamma\lb\sup_{\theta_0\le\theta\le\theta_1}\left|\frac{q_n^\ell-p_\theta^{T_0}}{p_\theta^{T_0}\lb1-p_\theta^{T_0}\rb}\right|\frac{\partial p_\theta^{T_0}}{\partial \theta}>x; q_n^\ell=0\rb\nonumber\\=&p_\gamma^{T_0}\mathbbm{1}_{\left\{ \sup_{\theta_0\le\theta\le\theta_1} {\frac{\partial p_\theta^{T_0}}{\partial \theta}}/{p_\theta^{T_0}}> x\right\}} 
+\lb1-p_\gamma^{T_0}\rb\mathbbm{1}_{\left\{ \sup_{\theta_0\le\theta\le\theta_1} {\frac{\partial p_\theta^{T_0}}{\partial \theta}}/{(1-p_\theta^{T_0})}> x \right\}}, \end{align}
Note that $\frac{\partial p_\theta^{T_0}}{\partial \theta}=\frac{\sqrt{T_0}}{\sqrt{2\pi}\sigma}\exp\lb -{(\lambda-\theta T_0)^2}/{(2\sigma^2T_0)}\rb\le \frac{\sqrt{T_0}}{\sqrt{2\pi}\sigma}$, and $0<p_{\theta_0}^{T_0}\le p_\theta^{T_0}\le p_{\theta_1}^{T_0}<1$, which lead to $$\sup_{\theta_0\le\theta\le\theta_1} {\frac{\partial p_\theta^{T_0}}{\partial \theta}}/{p_\theta^{T_0}}\le \frac{\sqrt{T_0}}{\sqrt{2\pi}\sigma}\frac{1}{p_{\theta_0}^{T_0}}\;\;\;\text{and}\;\; \sup_{\theta_0\le\theta\le\theta_1} {\frac{\partial p_\theta^{T_0}}{\partial \theta}}/{(1-p_\theta^{T_0})}\le \frac{\sqrt{T_0}}{\sqrt{2\pi}\sigma}\frac{1}{1-p^{T_0}_{\theta_1}}.$$ Hence, by letting $x_0=\max\left\{\frac{\sqrt{T_0}}{\sqrt{2\pi}\sigma}\frac{1}{p_{\theta_0}^{T_0}}, \frac{\sqrt{T_0}}{\sqrt{2\pi}\sigma}\frac{1}{1-p_{\theta_1}^{T_0}}\right\}$, we have $\Prob_\gamma\lb \sup_{\theta} |\nabla_{\theta}S_u(\theta, \gamma)| >x\rb=0<e^{-|\log x|^\eta}$ all $x>x_0, \eta>1$. Similarly, condition \eqref{A3b} holds as well. 
\end{itemize}
As a result, the performance of U-GSPRT can be characterized asymptotically by \eqref{U_error_perf}-\eqref{U_meansize_1}.
\ignore{\begin{align}
&\Prob_\theta\lb \left|\frac{q_n^\ell}{p_\gamma^{T_0}}-\frac{1-q_n^\ell}{1-p_\theta^{T_0}}\right|\frac{\partial p_\gamma^{T_0}}{\partial \gamma}>x\rb=\Prob_\theta\lb \sup_{\theta} \frac{1}{p_\theta^{T_0}}>\epsilon x\rb p_0^{T_0}
+\Prob_\theta\lb \sup_{\theta} \frac{1}{1-p_\theta^{T_0}}>\epsilon x\rb \lb1-p_0^{T_0}\rb, \end{align}}
\ignore{Note that $p_\gamma^{T_0}$ and $p^{T_0}_\theta$ in this case are strictly increasing functions of the unknown parameter $\gamma$ and $\theta$ respectively, i.e., $\frac{\partial p^q_\gamma}{\partial \gamma}>0$ and $\frac{\partial p^q_\theta}{\partial \theta}>0$. Then we have 
\begin{align}\label{A3a_q}
&\Prob_\gamma\lb \sup_{\theta}\left|\frac{q_n^\ell}{p_\theta^{T_0}}-\frac{1-q_n^\ell}{1-p_\theta^{T_0}}\right|\frac{\partial p_\theta^{T_0}}{\partial \theta}>x\rb\nonumber\\=&\Prob_\gamma\lb \sup_{\theta}\left|\frac{q_n^\ell-p_\theta^{T_0}}{p_\theta^{T_0}\lb1-p_\theta^{T_0}\rb}\right|\frac{\partial p_\theta^{T_0}}{\partial \theta}>x\rb\nonumber\\=&\Prob_\gamma\lb\sup_{\theta}\left|\frac{q_n^\ell-p_\theta^{T_0}}{p_\theta^{T_0}\lb1-p_\theta^{T_0}\rb}\right|\frac{\partial p_\theta^{T_0}}{\partial \theta}>x; q_n^\ell=1\rb+\Prob_\gamma\lb\sup_{\theta}\left|\frac{q_n^\ell-p_\theta^{T_0}}{p_\theta^{T_0}\lb1-p_\theta^{T_0}\rb}\right|\frac{\partial p_\theta^{T_0}}{\partial \theta}>x; q_n^\ell=0\rb\nonumber\\=&\Prob_\gamma\lb \sup_{p_\theta^{T_0}} \frac{1}{p_\theta^{T_0}}>\epsilon x\rb p_\gamma^{T_0}
+\Prob_\gamma\lb \sup_{p_\theta^{T_0}} \frac{1}{1-p_\theta^{T_0}}>\epsilon x\rb \lb1-p_\gamma^{T_0}\rb, \end{align}
where $\epsilon=1/\lb\frac{\partial p_\theta^{T_0}}{\partial \theta}\rb>0$. Note that $0<p_\theta^{T_0}<1$, thus we can always find sufficiently large $x_0$ such that for all $x>x_0$, both terms in \eqref{A3a_q} vanish. Similarly, we can prove the same for \eqref{A3b}. Thus \eqref{A3a}-\eqref{A3b} are satisfied, and Proposition \ref{C_GSPRT_perf} can be applied.
}

Next, we solve for the constrained MLE of the unknown parameter up to $n$th transmssion period:
\begin{align}\label{ex1_MLE}
\hat{\theta}_n&=
\arg \max_{\theta\ge \theta_0} \quad r_0^n \log \lb1-p^{T_0}_\theta(\lambda)\rb+r_1^n \log p^{T_0}_\theta\lb\lambda\rb\nonumber\\
&=\arg \max_{\theta\ge \theta_0}\quad r_0^n\log{\Phi\lb\frac{\lambda-\theta T_0}{\sigma\sqrt{T_0}}\rb}+r^n_1\log{\left(1-\Phi\lb\frac{\lambda-\theta T_0}{\sigma\sqrt{T_0}}\rb\right)},
\end{align}
where $r_0^n$ and $r_1^n$ represent the number of received ``$-1$'' and ``$+1$'' respectively among the first received $n$ bits.  By noting that the objective in \eqref{ex1_MLE} is a concave function of $\theta$, we can invoke the optimality condition and find the MLE as $\hat{\theta}_n=\mathcal{E}\left(\lambda-\Phi^{-1}\left(\frac{r_0^n}{r^n_0+r^n_1}\right)\sigma/\sqrt{T_0}, \theta_0, \theta_1\right)$.
\ignore{\begin{align}
\left(\frac{r^n_0}{\Phi\left(\frac{\lambda-\theta T_0}{\sigma\sqrt{T_0}}\right)}-\frac{r^n_1}{1-\Phi\left(\frac{\lambda-\theta T_0}{\sigma\sqrt{T_0}}\right)}\right)\frac{\partial \Phi\left(\frac{\lambda-\theta T_0}{\sigma\sqrt{T_0}}\right)}{\partial \theta}=\lambda/T_0\ge 0,\quad \lambda \left(\theta_0-\theta\right)=0, \quad \theta_0\le \theta\le \theta_1,
\end{align}} 

In the simulation experiment, we set the algorithm parameters as follows. The noise variance is normalized as one, i.e, $\sigma^2=1$. The parameter interval is $\theta\in [0.4, 2]$. The U-GSPRT is implemented in two settings, i.e., the inter-communication period $T_0=10$ and $T_0=1$ respectively. The expected inter-communication period for the level-triggered sampling scheme is fixed as approximately $\E T_0\approx 10$ by adjusting the local thresholds $\{a, b\}$. In both cases, the binary quantizer in  the minimax sense, i.e., the threshold that solves \eqref{opt_la}, is found to be $\lambda/T_0\approx 0.32$. 
\ignore{
\begin{table*}
\centering
\caption{}
  \begin{tabular}{ | c | c | c | c | c | }
    \hline
    \phantom{1} & $D\lb f_\theta||h_0\rb$ & $D\lb h_0||f_\theta\rb$ & $D\lb p^{T_0}_\theta||p^{T_0}_0\rb/T_0$ & $D\lb p^{T_0}_0||p^{T_0}_\theta\rb/T_0$ \\ \hline
    $T_0=1$ & $[0.6, 0.2, 0.2]$ & $1.6$ & $[0.1, 0.4, 0.5]$  & $2.4$  \\ \hline
    $T_0=10$ & $[0.1, 0.4, 0.5]$  & $2.4$ & $[0.6, 0.2, 0.2]$ & $1.6$ \\ \hline
   \end{tabular}
\end{table*}
}
\begin{figure}
\centering
{\includegraphics[width=0.99\textwidth]{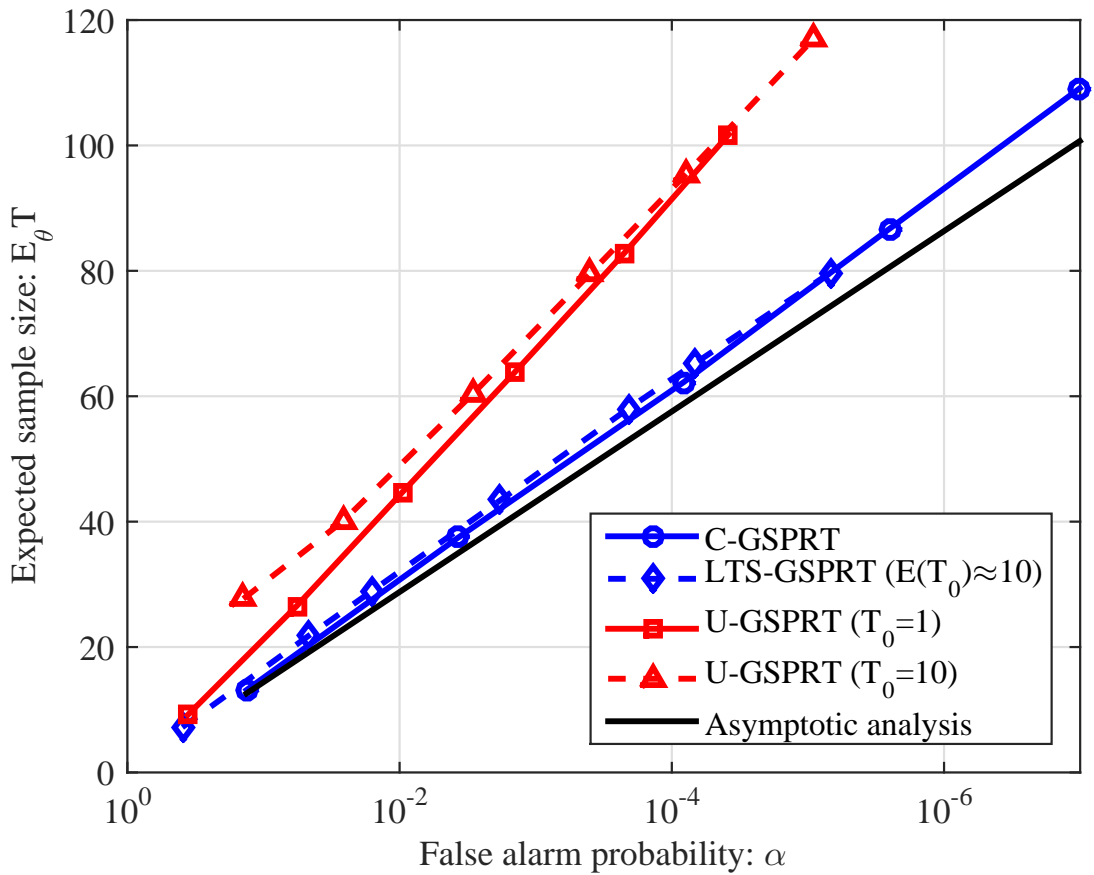}}
\caption{Expected samples versus false alarm probability $\alpha$.
}\label{fig1}
\end{figure}
\begin{figure}
\centering
{\includegraphics[width=0.99\textwidth]{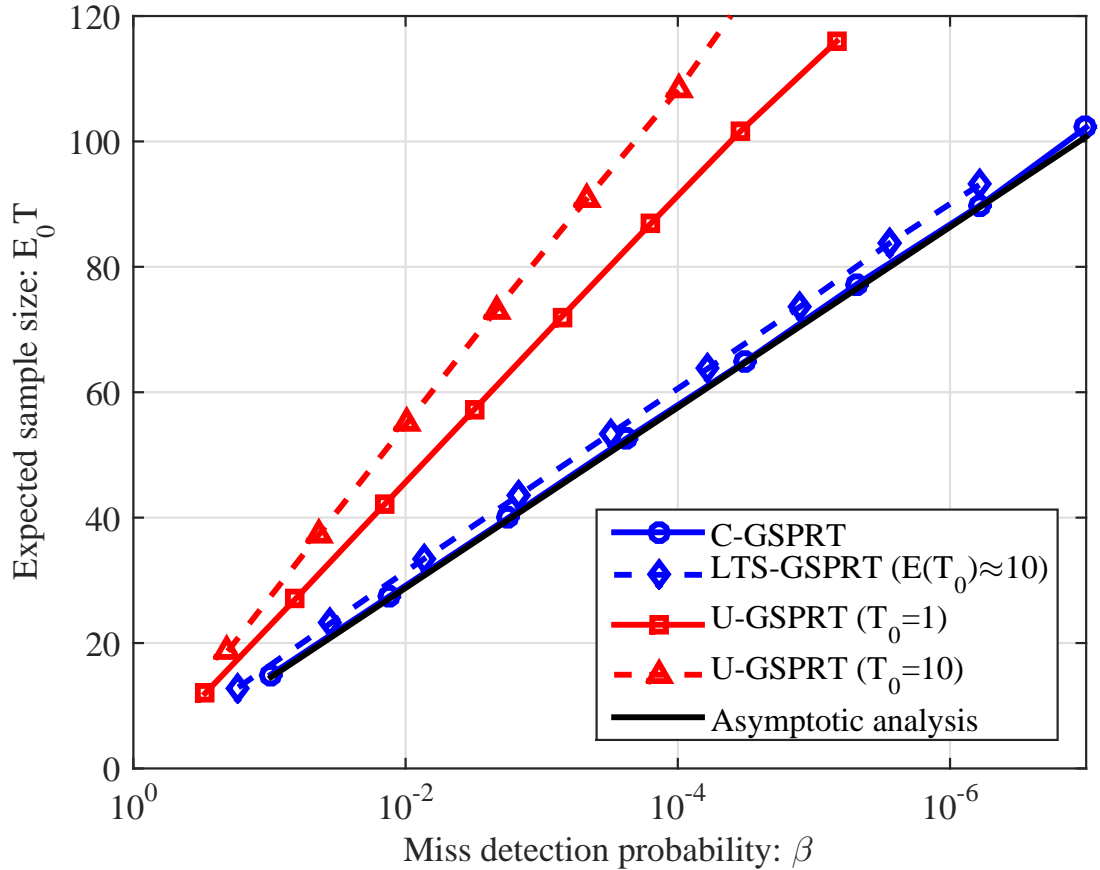}}
\caption{
Expected sample size versus miss detection probability $\beta$.}\label{fig2}
\end{figure}
\begin{figure}
\centering
{\includegraphics[width=0.99\textwidth]{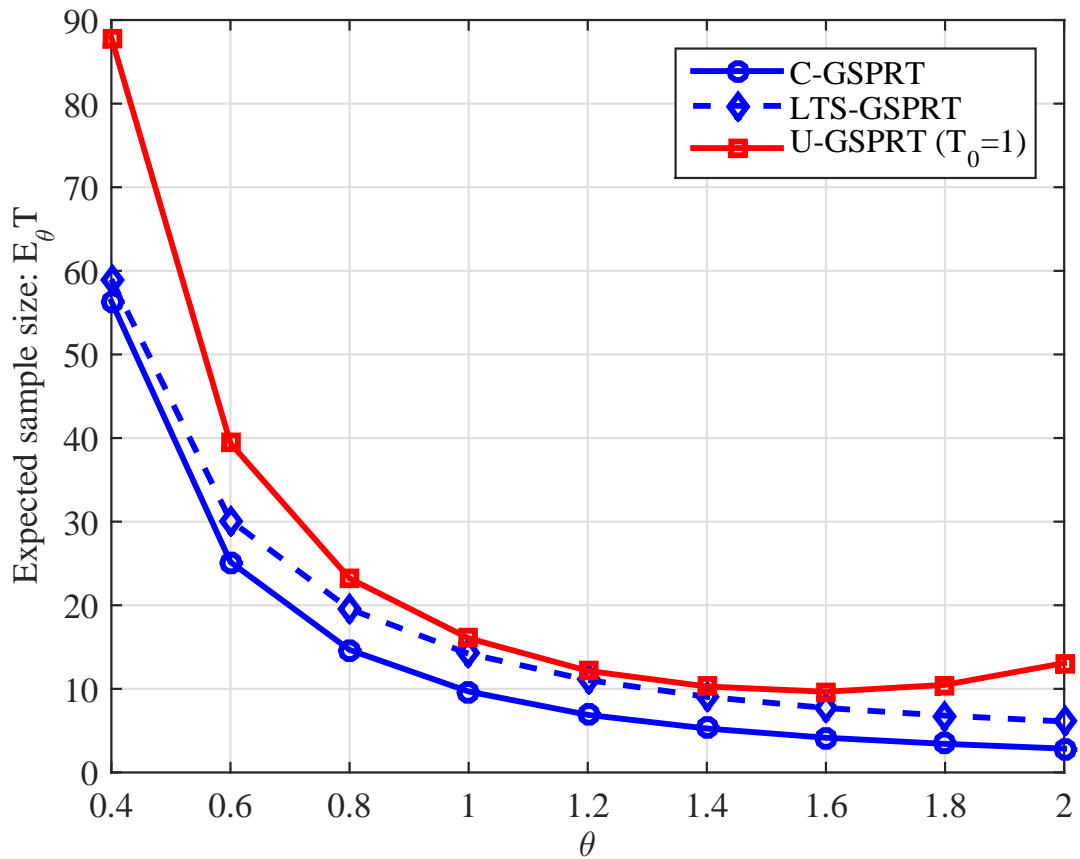}}
\caption{
Expected sample size versus varying parameter values.}\label{fig3}
\end{figure}
\begin{figure}
\centering
{\includegraphics[width=0.96\textwidth]{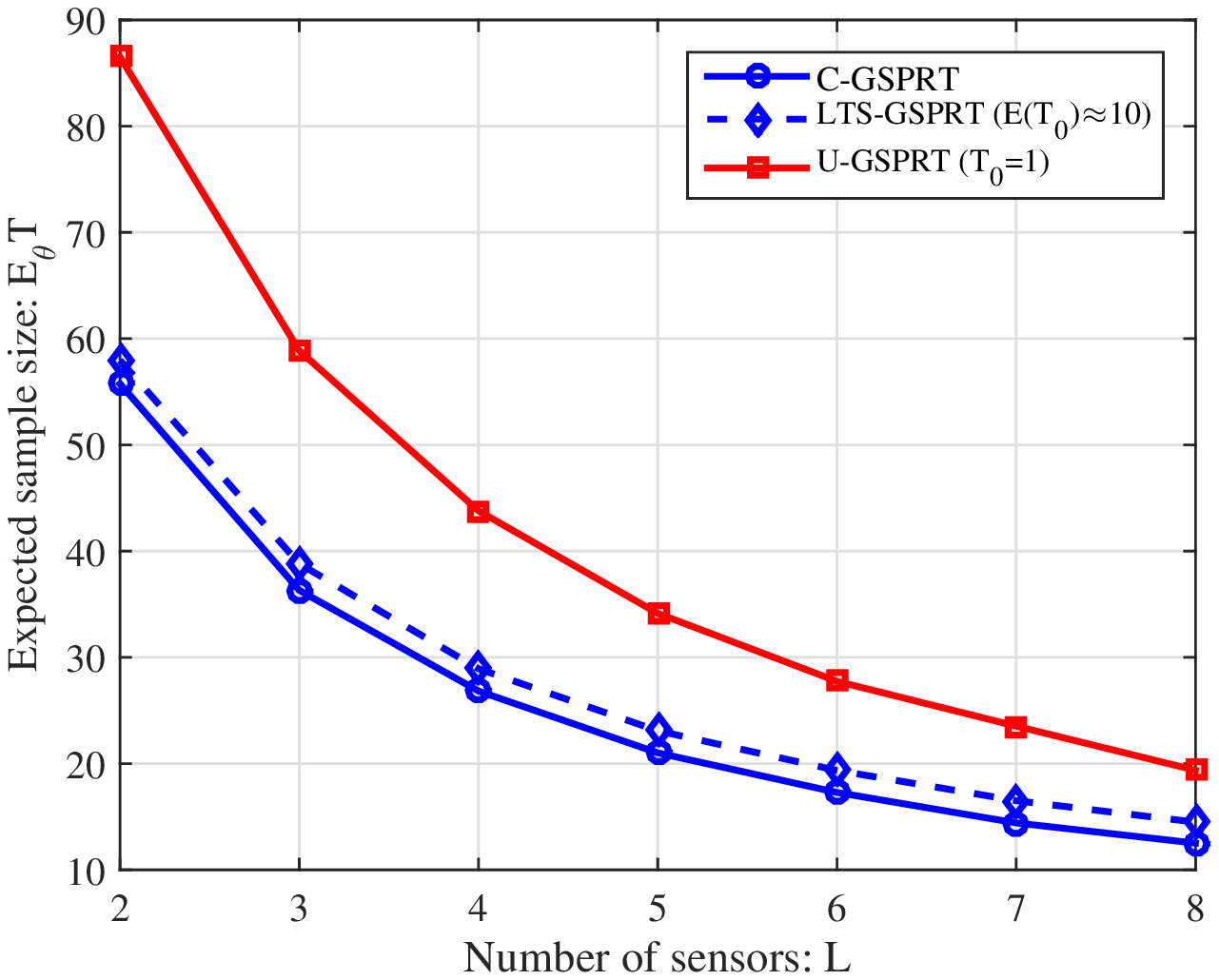}}
\caption{Expected sample size versus number of sensors.
}\label{fig4}
\end{figure}

In Figs. \ref{fig1}-\ref{fig2}, the performances of C-GSPRT, U-GSPRT and LTS-GSPRT are examined based on a two-sensor system. Specifically, Fig. \ref{fig1} depicts the expected sample size under the alternative hypothesis (with $\theta=0.4$) as a function of the false alarm probability, with the miss detection probability equal to $\beta\approx 10^{-4}$. Fig. \ref{fig2} depicts the expected sample size under the null hypothesis as a function of the miss detection probability, with the false alarm probability equal to $\alpha\approx 10^{-4}$. In these two figures, the black solid lines correspond to the following asymptotic formulas respectively (cf.  \eqref{asym_c}-\eqref{asym_d} without the $o(\cdot)$ terms), $$\E_\theta \T =\frac{-\log\alpha}{D\lb f_\theta||h_0\rb L}=\frac{-\log\alpha}{\theta^2L/2}, \quad \E_0 \T =\frac{-\log\beta}{\inf_\theta D\lb h_0||f_\theta\rb L}=\frac{-\log\beta}{\theta_0^2 L/2}.$$ {Note that since the true parameter in the experiment is $\theta_0=0.4$, $\inf_\theta D(h_0||f_\theta)=D(h_0||f_{\theta})$, the black-solid lines in Figs. \ref{fig1}-\ref{fig2} also correspond to the performance of SPRT for the simple null versus simple alternative test.} As expected, both C-GSPRT and LTS-GSPRT align closely with the asymptotic analysis. Notably, LTS-GSPRT only sacrifices a fractional sample-size compared to C-GSPRT while yielding substantially lower overhead through low-frequency one-bit communication. Figs. \ref{fig1}-\ref{fig2} also clearly show  that U-GSPRT diverges from C-GSPRT and LTS-GSPRT by an order of magnitude due to the smaller value of the KL divergence (i.e., $D\lb f_\theta||h_0\rb=0.08>D\lb p^{10}_\theta||p^{10}_0\rb/{10}\approx D\lb p^{1}_\theta||p^{1}_0\rb\approx 0.051$ and $\inf_\theta D\lb h_0||f_\theta\rb=0.08 >\inf_\theta D\lb p^{1}_0||p^{1}_\theta\rb\approx 0.050>\inf_\theta D\lb p^{10}_0||p^{10}_\theta\rb/{10}\approx 0.042$). Note that we also plot the performance of U-GSPRT for $T_0=1$ that corresponds to a binary quantization at every instant. It is seen that even with ten times more frequent communication to the fusion center, U-GSPRT is still outperformed by LTS-GSPRT substantially. 


Fig. \ref{fig3} illustrates the performances of C-GSPRT, U-GSPRT, LTS-GSPRT for varying parameter values. Note that all algorithms are implemented without this  knowledge, hence this figure shows how they adapt to different parameter values, which is a critical performance indicator for composite test. The error probabilities are fixed at $\alpha\approx 2\times 10^{-4}, \beta\approx 10^{-4}$. As $\theta$ varies from $0.4$ to $2$, the fusion center samples faster from the sensors, i.e., $\E_\theta (T_0)\approx 10\to 1.5$, due to the embedded adaptive mechanism. Meanwhile, U-GSPRT with the best time resolution $T_0=1$ is examined. 
It is clearly shown in Fig. \ref{fig3} that LTS-GSPRT is able to align with C-GSPRT closely and consistently outperforms U-GSPRT over all parameter values. Again, LTS-GSPRT results in the lowest communication overhead among these three tests. 

Fig. \ref{fig4} further examines the centralized and decentralized algorithms under  different number of sensors. The error probabilities are fixed at $\alpha\approx 2\times 10^{-4}, \beta\approx 10^{-4}$. Clearly, using more sensors brings down the sample size given a target accuracy. It is seen that, for a reasonable number of sensors in practice, e.g., eight sensors, LTS-GSPRT stays close to the centralized scheme and consistently exhibits smaller sample size compared to the uniform sampling based decentralized scheme. 

\subsection{Collaborative Sequential Spectrum Sensing}
In this subsection, we consider the collaborative sequential spectrum sensing in cognitive radio systems. To cope with the ever-growing number of mobile devices and the scarce spectrum resource, the emerging cognitive radio systems enable the  secondary users to quickly identify the idle frequency band for opportunistic communications. Moreover,  secondary users can collaborate to increase their spectrum sensing speed. Specifically, if the target frequency band is occupied by a primary user, the received signal by the $\ell$th secondary user can be written as 
\begin{align}
y_t^\ell=h_t^\ell s_t+e_t^\ell
\end{align}
where $h_t^\ell\sim \mathcal{N}\lb 0, 1 \rb$ is the normalized fading channel gain between the primary user and the $\ell$th secondary user, independent of the noise $e_t^\ell$ and $s_t$ is the unknown signal transmitted by the primary user with energy $\E|s_t|^2$; otherwise if the target frequency band is available, secondary users only receive noise. To this end, the sequential spectrum sensing can be modelled as the following composite hypothesis testing problem \cite{Segura10}:
 \begin{align}
\begin{array}{ll}
\mathcal{H}_0: & y_t^\ell\sim \mathcal{N}\lb 0, \gamma\rb, \quad 0<\gamma_0\le \gamma\le \gamma_1,\quad \ell\in {\cal L},\; t=1, 2, \ldots,\\
\mathcal{H}_1: & y_t^\ell\sim \mathcal{N}\lb 0, \theta\rb,\quad \gamma_1<\theta_0\le \theta\le \theta_1\quad \ell\in {\cal L}, \; t=1, 2, \ldots,
\end{array}
\end{align}
where the parameter intervals $[\gamma_0, \gamma_1]$ and $[\theta_0,\theta_1]$ are prescribed by practitioners. 

We begin by verifying that the log-likelihood ratio of $y_t^\ell$, i.e., 
\begin{align}
S(\gamma, \theta)=\frac{1}{2}\lb\frac{|y_t^\ell|^2}{\gamma}-\frac{|y_t^\ell|^2}{\theta}\rb+\frac{1}{2}\log\frac{\gamma}{\theta},
\end{align}
satisfies the conditions $A1$-$A4$. While conditions $A2$-$A3$ are easily verified, conditions $A1$ and $A4$ can be checked as follows: 
\begin{itemize}
\item The KL divergences admit
\begin{align*}
&D\lb f_\theta||h_\gamma\rb=\frac{1}{2}\left(\frac{\theta}{\gamma}-1\right)+\frac{1}{2}\log \frac{\gamma}{\theta},\\
\text{and}\quad &D\lb h_\gamma || f_\theta\rb=\frac{1}{2}\left(\frac{\gamma}{\theta}-1\right)+\frac{1}{2}\log \frac{\theta}{\gamma},
\end{align*}
which are both decreasing functions of $\gamma$ and increasing functions of $\theta$.
Let $0<\varepsilon <\min\{D\lb h_{\gamma_1}||f_{\theta_0}\rb, D\lb h_{\gamma_1}||f_{\theta_0}\rb\}$, we have $\inf_{\gamma_0\le \gamma \le \gamma_1}D\lb f_\theta||h_\gamma\rb\ge  D\lb f_{\theta_0}||h_{\gamma_1}\rb> \varepsilon$ and $\inf_{\theta_0\le\theta\le\theta_1} D(h_\gamma||f_\theta)\ge D\lb h_{\gamma_1}||f_{\theta_0}\rb>\varepsilon$;
\item For \eqref{A3a}, let $x>\frac{1}{2\theta_0}>0$, then we have
\begin{align}
&\Prob_\gamma\lb \sup_{\theta_0\le\theta\le\theta_1} |\nabla_{\theta}S(\theta, \gamma)| >x\rb 
\nonumber\\=& \Prob_\gamma\lb \sup_{\theta_0\le\theta\le\theta_1} \frac{1}{2\theta^2}\left|{(y_t^\ell)^2}-\theta\right|>x\rb 
\nonumber\\\le & \Prob_\gamma\lb \sup_{\theta_0\le\theta\le\theta_1} \max\{\frac{1}{2\theta},\frac{(y_t^\ell)^2}{2\theta^2}\}>x\rb 
\nonumber\\= & \Prob_\gamma\lb  \frac{(y_t^\ell)^2}{2\theta_0^2}>x\rb\label{condition3_1}
\\= & 2\Phi\lb \frac{-\sqrt{2x}\theta_0}{\sqrt{\gamma}} \rb,  
\end{align}
where the inequality holds because $(y_t^\ell)^2\ge 0, \theta>0$ and $|\{(y_t^\ell)^2\}-\theta|\le \max\{(y_t^\ell)^2, \theta\}$, and \eqref{condition3_1} holds because $x>\frac{1}{2\theta_0}$. Again, since $\Phi(-\sqrt{2x}\theta_0/\sqrt{\gamma}) \sim e^{-x\theta_0^2/\gamma}$, we can always find a sufficiently large $x_0$ such that $x>|\log x|^\eta$ , or equivalently, $\Prob_\gamma\lb \sup_{\theta_0\le\theta\le\theta_1} |\nabla_{\theta}S(\theta, \gamma)| >x \rb\le e^{-|\log x|^\eta}$ for $x>x_0, \eta>1$. Similarly, we can show that \eqref{A3b} holds as well.
\end{itemize}

With $A1$-$A4$ satisfied, we proceed to employ the centralized and LTS-based  GSPRTs to solve the collaborative sequential spectrum sensing problem, which can be characterized asymptotically by Proposition 1 and Theorems 1-2. Particularly, the centralized LLR at the fusion center is evaluated as
\begin{align}
S_j^{k}(\gamma,\theta)&=\log \frac{\frac{1}{\theta^{L(k-j+1)/2}}\exp\lb -\frac{1}{2}\sum_{\ell=1}^L\sum_{t=j}^k\frac{|y_t^\ell|^2}{\theta}\rb}{\frac{1}{\gamma^{L(k-j+1)/2}}\exp\lb -\frac{1}{2}\sum_{\ell=1}^L\sum_{t=j}^k\frac{|y_t^\ell|^2}{\gamma}\rb}\nonumber\\&=\lb\frac{1}{2\gamma}-\frac{1}{2\theta}\rb\mathcal{W}_j^{k}+\frac{L(k-j+1)}{2}\log \frac{\gamma}{\theta}, \quad \mathcal{W}_j^{k}\triangleq \sum_{\ell=1}^L\sum_{t=j}^k |y_t^\ell|^2.
\end{align}
As such, the centralized MLE of the unknown parameters $\gamma$ and $\theta$ are easily obtained as
$\hat{\gamma}_j^k=\mathcal{E}\left(\mathcal{W}_j^k/\lb k-j+1\rb/L, \gamma_0, \gamma_1\right)$ and $\hat{\theta}_j^k=\mathcal{E}\left(\mathcal{W}_j^k/\lb k-j+1\rb/L, \theta_0, \theta_1\right)$. Then the centralized GSPRT given by \eqref{C_GSPRT_s}-\eqref{C_GSPRT_d} can be implemented based on the GLLR $\widetilde{S}_j^k= S_j^k\lb\hat{\gamma}_j^k, \hat{\theta}_j^k\rb$. In order to implement LTS-based GSPRT, the local LLR at sensor $\ell$ is 
\begin{align}\label{exa_2_localLLR}
S_j^{k, \ell}(\gamma, \theta)&=\lb\frac{1}{2\gamma}-\frac{1}{2\theta}\rb\mathcal{W}_j^{k, \ell}+\frac{k-j+1}{2}\log \frac{\gamma}{\theta}, \quad \mathcal{W}_j^{k, \ell}\triangleq \sum_{t=j}^k |y_t^\ell|^2.
\end{align}
Substituting $\hat{\gamma}_j^{k, \ell}=\mathcal{E}\left(\mathcal{W}_j^{k, \ell}/\lb k-j+1\rb, \gamma_0, \gamma_1\right)$ and $\hat{\theta}_j^{k, \ell}=\mathcal{E}\left(\mathcal{W}_j^{k, \ell}/\lb k-j+1\rb, \theta_0, \theta_1\right)$ into \eqref{exa_2_localLLR} gives local GLLR $\widetilde{S}_j^{k, \ell}(\hat{\gamma}_j^{k,\ell}, \hat{\theta}_j^{k,\ell})$, which is further plugged into  \eqref{LTS-Sampling}-\eqref{local_GLR} to run the LTS-GSPRT $\T_p$. 
To realize U-GSPRT for this problem, given the inter-communication period $T_0$, the sufficient statistic is found to be $\phi_j^{k, \ell}=\mathcal{W}_{j}^{k, \ell}$, which is defined in \eqref{exa_2_localLLR}, with different distributions under the null and alternative hypotheses:
\begin{align}
\mathcal{W}_{(n-1)T_0+1}^{nT_0, \ell}/\gamma\overset{\mathcal{H}_0}{\sim}\chi^2_{T_0}\lb 0\rb, \quad \mathcal{W}_{(n-1)T_0+1}^{nT_0, \ell}/\theta\overset{\mathcal{H}_1}{\sim} \chi^2_{T_0}\lb 0\rb.
\end{align}
Therefore, the binary quantizer for this problem is written as
\begin{align}\label{example2_quantizer}
q_n^\ell=\text{sign}\left(\mathcal{W}_{(n-1)T_0+1}^{nT_0,\ell}-\lambda\right), 
\end{align}
whose distribution is
\begin{align}\label{Up_2}
 p_x^{T_0}(\lambda)=1-\xi_{T_0}\lb \frac{\lambda}{x} \rb, \quad x\in [\gamma_0, \gamma_1]\cup [\theta_0, \theta_1],
\end{align}
where $\xi_k\lb x\rb$ is the CDF of the chi-squared distribution with degree of freedom $k$. By solving the maximum likelihood problem, it is straightforward to find the estimates of $\gamma$ and $\theta$ respectively as
\begin{align}
\hat{\theta}_n=\mathcal{E}\left(\frac{\lambda}{\xi_{T_0}^{-1}\lb\frac{r_0^n}{r^n_0+r^n_1}\rb}, \theta_0, \theta_1\right), \quad \hat{\gamma}_n=\mathcal{E}\left(\frac{\lambda}{\xi_{T_0}^{-1}\lb\frac{r_0^n}{r^n_0+r^n_1}\rb}, \gamma_0, \gamma_1\right).
\end{align}

Note that the log likelihood ratio of $q_n^\ell$ is the same as \eqref{U_LLR_1} with $p_x^{T_0}(\lambda)$ replaced by \eqref{Up_2}. Therefore, conditions $A1$ and $A4$ are verified by noting that 
\begin{itemize}
\item $p_\theta^{T_0}(\lambda)\neq p_\gamma^{T_0}(\lambda)$ for all $\theta\in [\theta_0,\theta_1]$ and $\gamma\in [\gamma_0, \gamma_1]$,  given any $\lambda$;
\item $\sup_{\theta_0\le\theta\le\theta_1} {\frac{\partial p_\theta^{T_0}}{\partial \theta}}/{p_\theta^{T_0}}$ and $\sup_{\theta_0\le\theta\le\theta_1} {\frac{\partial p_\theta^{T_0}}{\partial \theta}}/{1-p_\theta^{T_0}}$ are bounded, thus the same argument as in \eqref{A3a_q} applies. This is seen by recalling the density function of the chi-squared distribution, $$\frac{\partial p_x^{T_0}}{\partial x}=\frac{\lambda}{x^2}\xi'_{T_0}\lb\frac{\lambda}{x}\rb \le \lb\frac{\lambda}{x}\rb^{T_0/2}\frac{1}{2^{T_0/2}\Gamma(T_0/2)x},$$ with $x$ residing in a compact set, i.e., $x\in [\gamma_0, \gamma_1]\cup [\theta_0, \theta_1]$.
\end{itemize}
Then we can also  asymptotically characterize the performance of U-GSPRT in the sequential spectrum sensing problem by \eqref{U_error_perf}-\eqref{U_meansize_1}.


In the simulation experiment, the parameter intervals of interest are set as $\gamma\in [0.2, 1]$ and $\theta\in [2,5]$. We consider U-GSPRT with the best time resolution $T_0=1$, where the minimax quantizer $\lambda=\arg \max \min_{\theta} D\lb f_\theta||h_\gamma\rb \approx 3.8$. The expected inter-communication period for LTS-GSPRT is again set approximately as $\E T_0\approx 10$. 


In Fig. \ref{fig5}-\ref{fig6}, the performances of two-user C-GSPRT, U-GSPRT and LTS-GSPRT are examined with $\gamma=1, \theta=2$ in terms of the expected sample size (i.e., spectrum sensing speed) as a function of the false alarm probability and miss detection probability respectively (with $\beta\approx 10^{-4}$ in Fig. \ref{fig5} and $\alpha\approx 10^{-4}$ in Fig. \ref{fig6}). In both figures, the asymptotic optimality of LTS-GSPRT is clearly demonstrated as it aligns closely with C-GSPRT. In contrast, U-GSPRT diverges significantly from C-GSPRT and LTS-GSPRT due to the smaller values of the KL divergence.
Furthermore, Fig. \ref{fig7} compares the three sequential schemes for different parameter values and Fig. \ref{fig8} further depicts their performances with different number of collaborative secondary users with the error probabilities $\alpha\approx \beta\approx 10^{-4}$. Note that, although U-GSPRT sends local statistics to the fusion center every sampling instant, it is consistently outperformed by LTS-GSPRT where each user transmits the one-bit message only every ten sampling instants on average. More importantly, LTS-GSRPT only compromises a small amount of  expected sample size compared to the C-GSPRT while substantially lowering the communication overhead. In cognitive radio systems, such an advantage brought by LTS-GSPRT allows the secondary users to identify available spectrum resource  in a fast and economical fashion.
\begin{figure}
\centering
\includegraphics[width=0.99\textwidth]{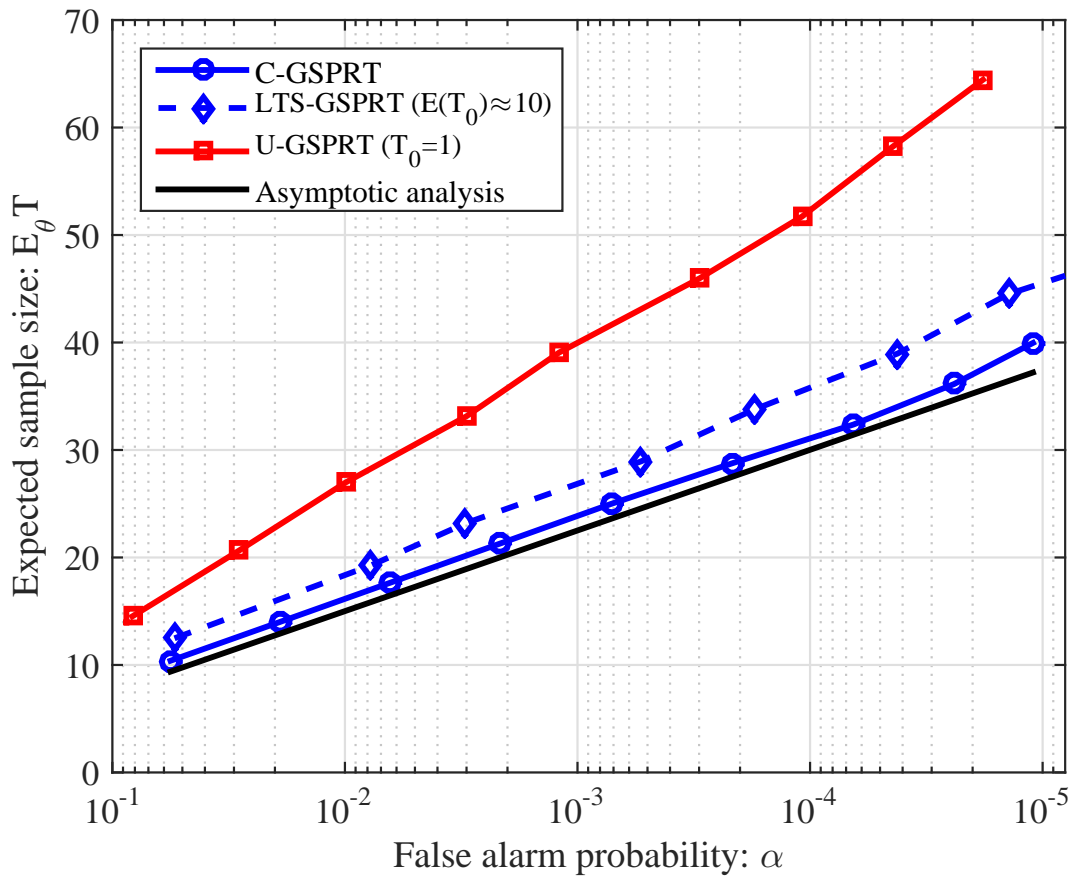}
\caption{Spectrum sensing speed versus false alarm probability $\alpha$.}\label{fig5}
\end{figure}
\begin{figure}
\centering
\includegraphics[width=0.99\textwidth]{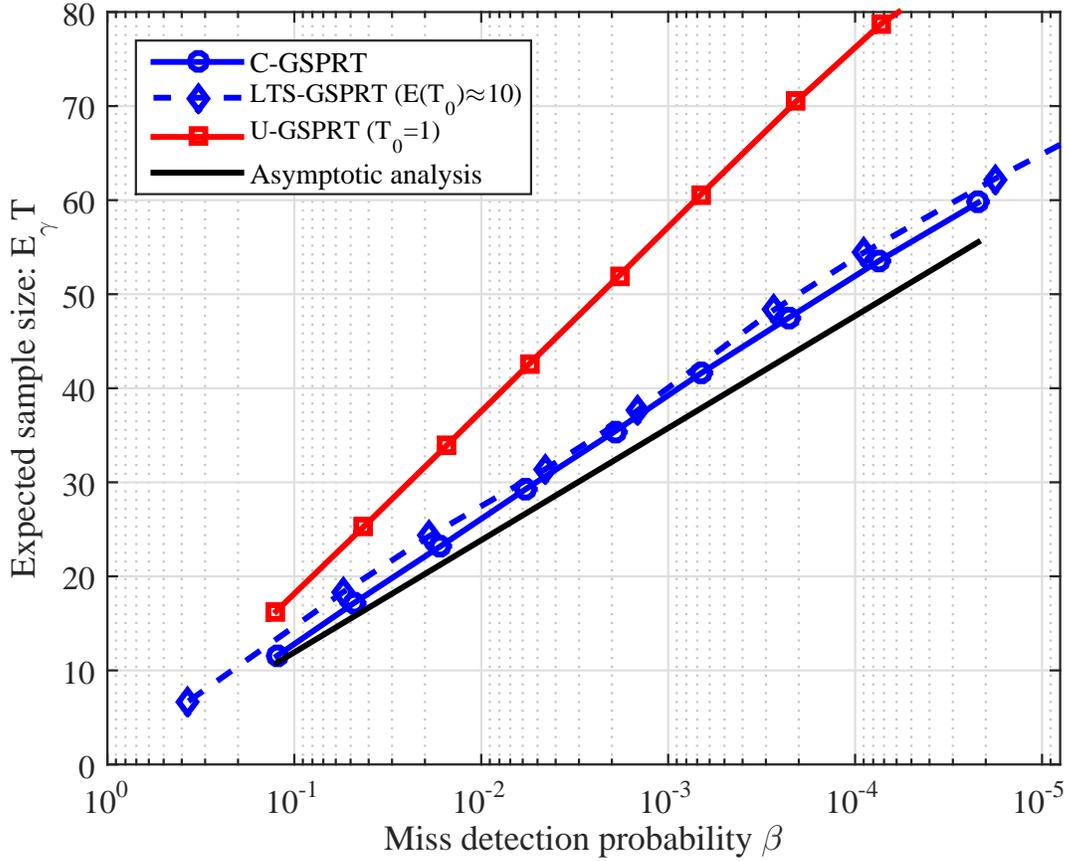}
\caption{Spectrum sensing speed versus miss detection probability $\beta$.}\label{fig6}
\end{figure}
\begin{figure}
\centering
\includegraphics[width=0.86\textwidth]{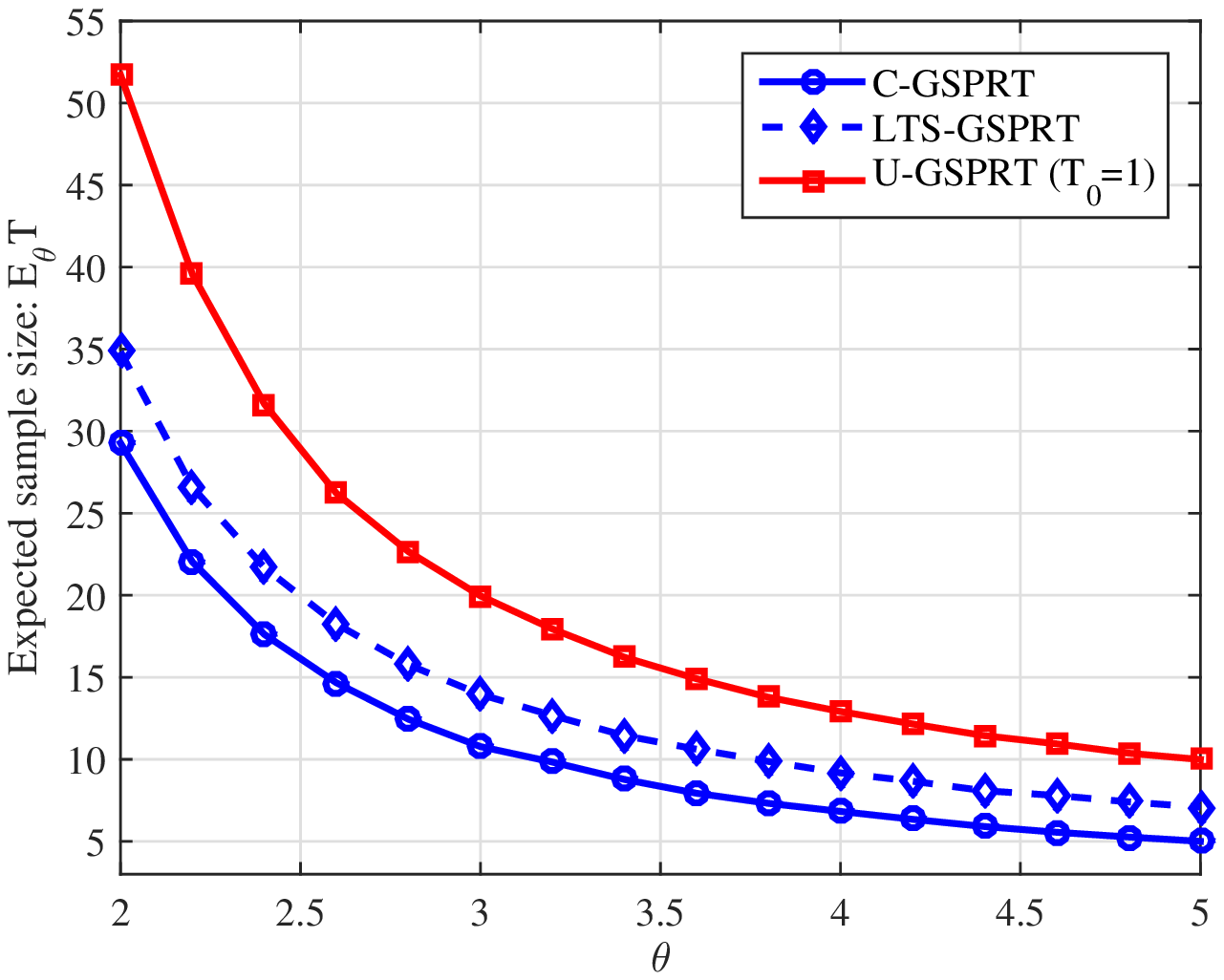}
\includegraphics[width=0.86\textwidth]{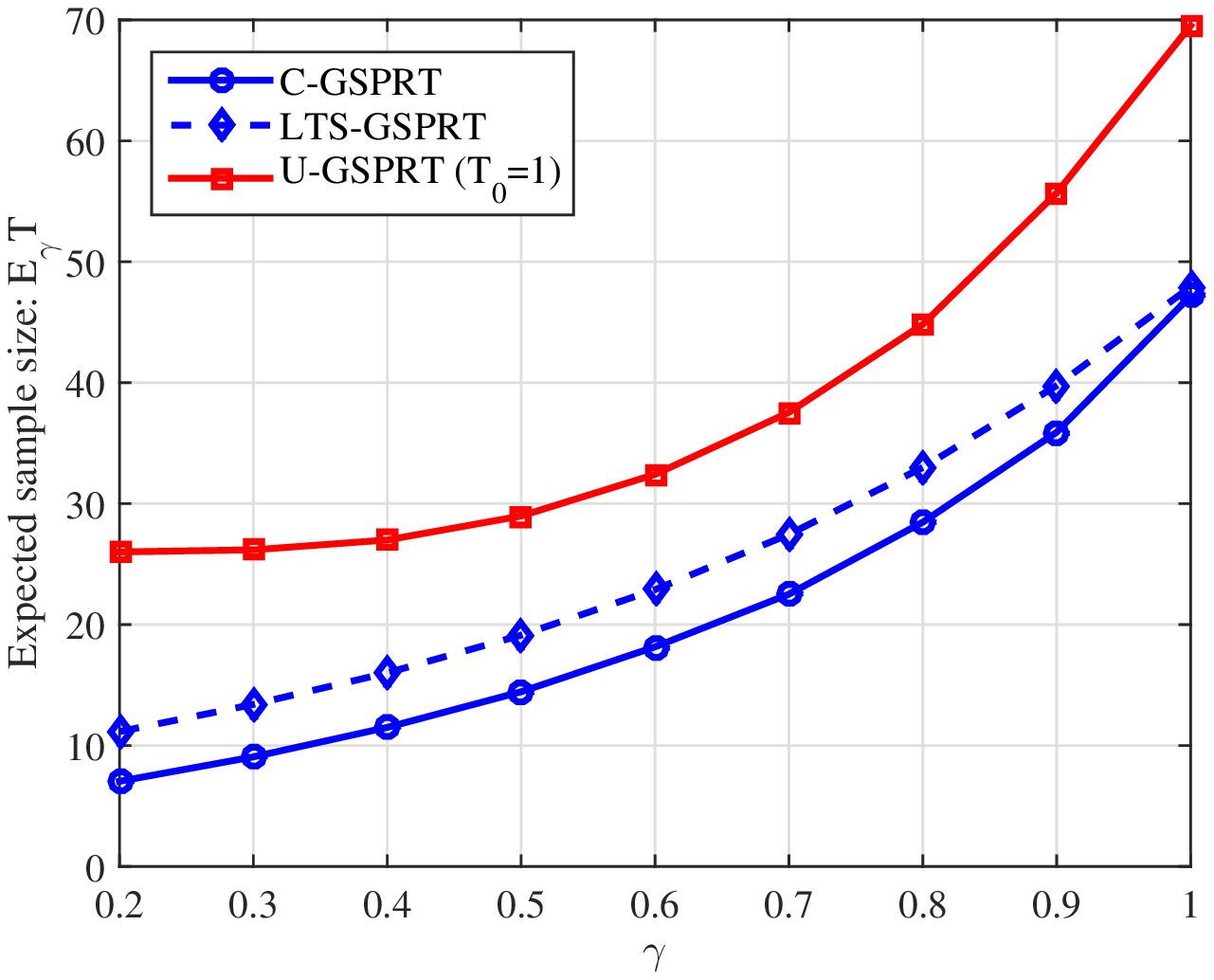}
\caption{Spectrum sensing speed versus different parameter values with and without the primary user.}\label{fig7}
\end{figure}
\begin{figure}
\centering
\includegraphics[width=0.99\textwidth]{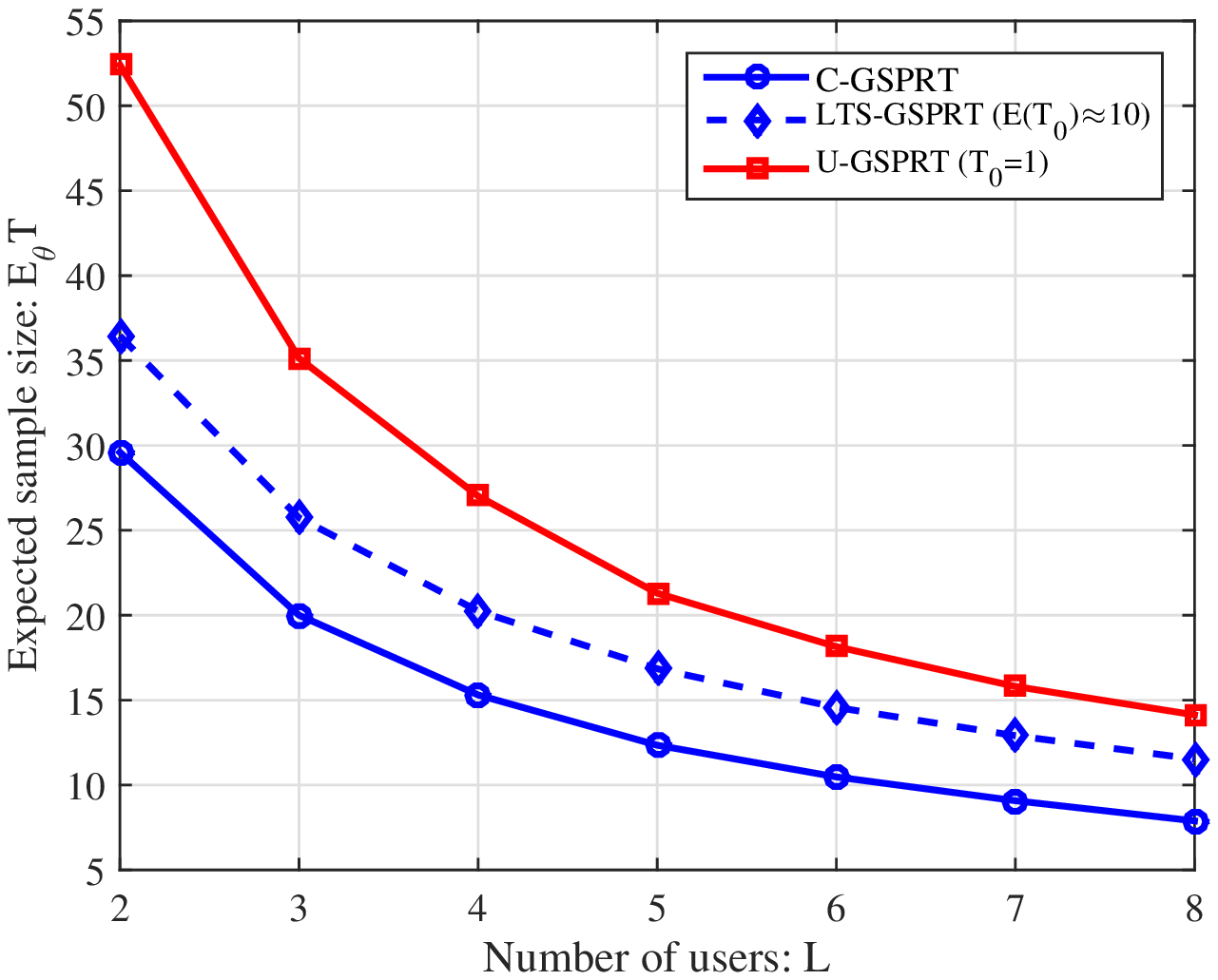}
\caption{Spectrum sensing speed versus different number of collaborating secondary users.}\label{fig8}
\end{figure}
\section{Conclusions}
This work has investigated the sequential composite hypothesis test based on data samples from multiple sensors. We have first introduced the GSPRT as an asymptotically optimal centralized scheme that serves as a benchmark for all decentralized schemes. Next a decentralized sequential test based on conventional uniform sampling and one-bit quantization has been studied, which is shown to be strictly suboptimal due to the loss of time resolution and coarse quantization. Then, by employing the level-triggered sampling, we have proposed a novel decentralized sequential scheme, where sensors repeatedly run local GSPRT and report their decisions to the fusion center asynchronously, and an approximate GSPRT based on the local decisions is performed at the fusion center. The LTS-based GSPRT significantly lowers the communication overhead through low-frequency one-bit communication, and is easily implemented both at sensors and the fusion center. Most importantly, we have shown that the proposed LTS-based decentralized scheme achieves the asymptotical optimality as the local thresholds and the global thresholds grow large at different rates. Finally, extensive numerical results have corroborated the theoretical results and demonstrated the superior performance of the proposed method. 
\section*{Appendix}
\subsection{Proof of Theorem 1}
We first introduce the following result as an extension to the Wald's identity, that can be found in \cite[Lemma 3]{Fellouris11}.
\begin{lemma}\label{lemma:ext_Wald}
Let $\{t_{n}^{\ell}\}$ be defined by \eqref{LTS-Sampling}. Consider a sequence $\{\psi_{n}^\ell\}$ of i.i.d. random variables where each $\psi_{n}^\ell$ is a function of the samples $y^\ell_{t^\ell_{n-1}+1}, \ldots, y^\ell_{t^\ell_{n}}$ acquired by  sensor $\ell$ during its $n$th inter-communication period. Then the following equality holds:
\begin{align}\label{ext_Wald}
\E_x\left(\sum_{n=1}^{N_{\T}^\ell+1}{\psi}_{n}^\ell\right)=\E_x\left({\psi}_{n}^\ell\right)\E_x\left(N_{\T}^\ell+1\right), \quad x\in \Gamma\cup\Theta.
\end{align}
\end{lemma}
Here, \eqref{ext_Wald} differs from the standard Wald's identity because $N^\ell_\T+1$ is no longer a stopping time adapted to $\{{\psi}^\ell_{n}\}$. Next we proceed to analyse the expected sample size under level-triggered sampling. Since the proof is concentrated on the LTS-based decentralized scheme only, we use $\T$ for  $\T_p$ (cf. \eqref{eq:D-GSPRT}) for notational simplicity.

\begin{proof}[Proof of Theorem 1]
Note that the global statistic in \eqref{eq:bS} can be rewritten as
\begin{align}
\widetilde{V}_t=&\sum_{\ell=1}^L\sum_{n=1}^{N_t^\ell}\tilde{v}^\ell_{n}\nonumber\\=&\sum_{\ell=1}^L\left(\sum_{n=1}^{N_t^\ell+1}\tilde{v}^\ell_{n}-\tilde{v}^\ell_{N_t^\ell+1}\right)\nonumber\\=&\sum_{\ell=1}^L\sum_{n=1}^{N_t^\ell+1}\tilde{v}^\ell_{n}-\sum_{\ell=1}^L\tilde{v}^\ell_{N_t^\ell+1}.
\end{align}
Thus, invoking Lemma \ref{lemma:ext_Wald}, we have
\begin{align}\label{eq:1}
\E_x\left(\widetilde{V}_\T\right)&=\sum_{\ell=1}^L\E_x\left(N_\T^\ell+1\right)\E_x(\tilde{v}^\ell_{n})-\sum_{\ell=1}^L\E_x\left(\tilde{v}^\ell_{N_\T^\ell+1}\right), \quad \;\; x\in\Gamma\cup\Theta.
\end{align}
Denote the inter-communication period $\tau^\ell_{n}\triangleq t^\ell_{n}-t^\ell_{n-1}$.   
Further define $R_\ell\triangleq \sum^{N_\T^\ell+1}_{n=1}\tau^\ell_{n}-\T\ge 0$, by noting that $\T\le\sum^{N_\T^\ell+1}_{n=1}\tau^\ell_{n}$. As a result, we can write down the following equality for each sensor:
\begin{align}\label{eq:3}
\E_x\lb\T+R_\ell\rb=\E_i\lb\sum^{N_\T^\ell+1}_{n=1}\tau^\ell_{n}\rb=\E_x\lb N_\T^\ell+1\rb\E_x\lb \tau_{n}^\ell\rb, \quad \ell=1, \ldots, L.
\end{align}
Combining \eqref{eq:3} and \eqref{eq:1} yields
\begin{align}\label{eq:2}
\E_x\lb\widetilde{V}_\T\rb&=\sum_{\ell=1}^L\frac{\E_x\lb\T+R_\ell\rb}{\E_i\lb\tau_{n}^\ell\rb}\E_i(\tilde{v}^\ell_{n})-\sum_{\ell=1}^L\E_x\left(\tilde{v}_{N_\T^\ell+1}\right)
\nonumber\\&=\E_x\lb\T\rb\sum_{\ell=1}^L\frac{\E_x(\tilde{v}^\ell_{n})}{\E_x\lb\tau_{n}^\ell\rb}+\sum_{\ell=1}^L\lb\E_x(R_\ell)\frac{\E_x(\tilde{v}^\ell_{n})}{\E_x\lb\tau_{n}^\ell\rb}-\E_x\lb\tilde{v}_{N_\T^\ell+1}\rb\rb.
\end{align}
Furthermore, according to Proposition \ref{C_GSPRT_perf}, we have
\begin{align}
&\E_\theta\lb\tilde{v}_{n}^\ell\rb\sim {a}\sim\E_\theta\lb\tau_{n}^\ell\rb{\inf_\gamma D\lb f_\theta||h_\gamma\rb}, \quad \text{as}\;\; \widetilde{\alpha}\to 0,\\
\text{and}\quad &\E_\gamma\lb\tilde{v}_{n}^\ell\rb\sim -{b}\sim -\E_\gamma\lb\tau_{n}^\ell\rb{\inf_\theta D\lb h_\gamma||f_\theta\rb}, \quad \text{as}\;\; \widetilde{\beta}\to 0.\label{eq:4}
\end{align}
Considering the sensor samples under hypothesis $\mathcal{H}_1$, \eqref{eq:2} becomes
\begin{align}
\E_\theta\lb\widetilde{V}_\T\rb=\E_\theta\lb\T\rb\sum_{\ell=1}^L\inf_\gamma D\lb f_\theta||h_\gamma \rb-\sum_{\ell=1}^L\underbrace{\lb\E_\theta\lb\tilde{v}_{N_\T^\ell+1}\rb-\E_\theta(R_\ell)\inf_\gamma D\lb f_\theta || h_\gamma\rb\rb}_{\mathcal{R}_\theta^\ell},
\end{align}
which leads to
\begin{align}
\E_\theta \lb\T\rb = \frac{\E_\theta\lb\widetilde{V}_\T\rb+\sum_{\ell=1}^L\mathcal{R}^\ell_\theta}{\inf_\gamma{D}\lb f_\theta||h_\gamma\rb L}\sim \frac{A+\sum_{\ell=1}^L\mathcal{R}^\ell_\theta}{\inf_\gamma{D}\lb f_\theta||h_\gamma\rb L}, \quad \text{as}\;\; \widetilde{\alpha}\to 0, \;\widetilde{\beta}\to 0.
\end{align}
Similarly, substituting \eqref{eq:4} into \eqref{eq:2} gives
\begin{align}
\E_\gamma\lb\widetilde{V}_\T\rb=-\E_\gamma\lb\T\rb\sum_{\ell=1}^L\inf_\theta D\lb h_\gamma||f_\theta\rb-\sum_{\ell=1}^L\underbrace{\lb\E_\gamma\lb\tilde{v}_{N_\T^\ell+1}\rb+\E_\gamma(R_\ell)\inf_\theta D\lb h_\gamma || f_\theta\rb\rb}_{\mathcal{R}_\gamma^\ell},
\end{align}
and the expected sample size under the null hypothesis is
\begin{align}
\E_\gamma \lb\T\rb = \frac{-\E_\gamma\lb\widetilde{V}_\T\rb-\sum_{\ell=1}^L\mathcal{R}^\ell_\gamma}{\inf_\theta{D}\lb h_\gamma||f_\theta\rb L} \sim \frac{B-\sum_{\ell=1}^L\mathcal{R}^\ell_\gamma}{\inf_\theta{D}\lb h_\gamma||f_\theta\rb L}\quad \text{as}\;\; \widetilde{\alpha}\to 0, \;\widetilde{\beta}\to 0.
\end{align}
We also have $\E_\theta\lb\widetilde{V}_\T\rb\to A,\; \E_\gamma\lb\widetilde{V}_\T\rb\to -B$, as $A, B \to \infty$ and ${a}=o\lb A\rb, {b}=o\lb B\rb$. Note that $\mathcal{R}_\theta^\ell$ and $\mathcal{R}_\gamma^\ell$ only depend on local thresholds $\{{b}, {a}\}$, which are of a lower order of $\{B, A\}$; therefore, we have proved the asymptotic formulas \eqref{mean_size_0} and \eqref{mean_size_1}.
\end{proof}
\subsection{Proof of Theorem 2}

The proof considers the asymptotic regime where $A/{a} \to\infty, B/{b} \to\infty$ and $\limsup {a}/{b} <\infty$. Again, $\T$ is used for  $\T_p$ for notational simplicity.
\begin{proof}
For simplicity of notations, we assume $L=2$ in the proof. When $L>2$, the proof is similar and is thus omitted. Thanks to the symmetry of type I and type II error probabilities, it is sufficient to compute the type I error probability. 
For any $\gamma\in\Gamma$, we consider the probability
\begin{equation}\label{errorprob}
\Prob_{{\gamma}}(\widetilde V_\T\geq  A).
\end{equation}
We first define the local discretized approximated generalized log-likelihood ratio process,
$$
\widetilde V_t^{(\ell)}= \sum_{n=1}^{N_t^{\ell}} {a} \mathbbm{1}_{\{u_n^{\ell}=1\}}-{b} \mathbbm{1}_{\{u_n^{\ell}=-1\}}, \quad  \ell=1,2,...,L.
$$
Then \eqref{errorprob} has the following upper bound
\begin{equation}\label{probofsum}
\Prob_{\gamma}(\widetilde V_\T\geq A) \leq \Prob_{\gamma}(\sup_{t} V_t\geq A) \leq \Prob_{\gamma}\Big( \sup_{t} \widetilde V^{(1)}_t +\sup_{t} \widetilde V^{(2)}\geq A\Big).
\end{equation}
The first inequality is due to the definition of $\T$, and the second inequality is because $\sup V_t \leq \sum_{\ell=1}^L \sup_t \widetilde V_t^{(\ell)}$.
We proceed to split the last probability in \eqref{probofsum} into error probabilities  detected by the local sensors. Let $\varepsilon$ be an arbitrary positive constant, then
\begin{eqnarray*}
&& \Prob_{\gamma}\Big( \sup_{t} \widetilde V^{(1)}_t +\sup_{t} \widetilde V^{(2)}_t\geq A\Big)\\
&\leq& \sum_{k=1}^{\lfloor 1/\varepsilon \rfloor } 
\Prob_{{\gamma}}\Big(
k\varepsilon A\leq \sup_t \widetilde V_t^{(1)}\leq (k+1)\varepsilon A, \sup_t \widetilde V_t^{(2)}\geq (1-(k-1)\varepsilon) A
\Big)\\
&&+ \Prob_{\gamma}\Big(\sup_{t} \widetilde V^{(1)}_t\leq \varepsilon A, \sup_{t} \widetilde V^{(2)}_t\geq (1-\varepsilon) A\Big).\\
\end{eqnarray*}
Note that the stochastic processes $\{V^{(1)}_{t}:t>0\}$ and $\{V^{(2)}_t:t>0\}$ are independent and identically distributed, so the right-hand side of the above inequality equals to
\begin{eqnarray*}
&&\sum_{k=1}^{\lfloor 1/\varepsilon \rfloor } 
\Prob_{{\gamma}}\Big(
k\varepsilon A \leq \sup_t \widetilde V_t^{(1)}\leq (k+1)\varepsilon A\Big) \Prob_{{\gamma}}\Big(\sup_t \widetilde V_t^{(1)}\geq (1-(k-1)\varepsilon) A
\Big)\\
&&+\Prob_{\gamma}(\sup_{t} \widetilde V^{(1)}_t\leq \varepsilon A)\Prob_{{\gamma}}\Big( \sup_{t} \widetilde V^{(1)}_t\geq (1-\varepsilon) A\Big),
\end{eqnarray*}
which can be further bounded above by
\begin{equation}\label{eq:upperbound}
\sum_{k=1}^{\lfloor 1/\varepsilon \rfloor } 
\Prob_{{\gamma}}\Big(
 \sup_t \widetilde V_t^{(1)}\geq k\varepsilon A\Big) \Prob_{{\gamma}}\Big(\sup_t \widetilde V_t^{(1)}\geq (1-(k-1)\varepsilon) A
\Big)+\Prob_{{\gamma}}\Big( \sup_{t} \widetilde V^{(1)}_t\geq (1-\varepsilon) A\Big).
\end{equation}
For each $k$ such that $1\leq k\leq \lfloor\frac{1}{\varepsilon}\rfloor$, we have $\varepsilon\leq k\varepsilon\leq 1$ and $(1-(k-1)\varepsilon)=1-k\varepsilon+\varepsilon$.
Consequently, \eqref{eq:upperbound} can be further bounded above by
\begin{equation}\label{totalbound}
\varepsilon^{-1} \sup_{\rho \in [\varepsilon,1] } \Prob_{{\gamma}}\Big( \sup_t \widetilde V^{(1)}_t\geq \rho A\Big)\Prob_{{\gamma}}\Big(\sup_t \widetilde V^{(1)}_t\geq (1-\rho+\varepsilon)A \Big).
\end{equation}
Then we use the following lemma whose proof is given below 
to complete the proof of Theorem 2.
\begin{lemma}\label{localerror}
For $\varepsilon>0$ and $\rho\geq \varepsilon$,
$$
\Prob_{{\gamma}}\Big(\sup_t \widetilde V^{(1)}_t\geq \rho A\Big) \leq e^{-(1+o(1))\rho A} \mbox{ as } A\to \infty.
$$
The above limit is uniform with respect to $\rho$ and $\gamma$.
\end{lemma}
Applying Lemma~\ref{localerror} to \eqref{totalbound} gives the result in Theorem 2.
\end{proof}

\begin{proof}[Proof of Lemma~\ref{localerror}]\label{prooflocalerror}
To start with, we write $\widetilde V_t^{(1)}$ in terms of the sum of i.i.d. variables,
$$
\widetilde V^{(1)}_t=\sum_{n=1}^{N_t} {a} \mathbbm{1}_{ \{u_n^1=1\} }-{b}  \mathbbm{1}_{ \{u_n^1=-1\} }.
$$
Therefore, the event $\{\sup_t \widetilde V^{(1)}_t\geq \rho A\}$ is the same as the event 
$$
\Big\{\sup_{N} \sum_{n=1}^N  Y_n\geq \rho A\Big\},
$$
where
$$
Y_n= {a} \mathbbm{1}_{ \{u_n^1=1\} }-{b}  \mathbbm{1}_{ \{u_n^1=-1\} }, \quad n=1,2,...
$$
The above event is further equivalent with the event 
$$
\{N^*<\infty\},
$$
where $N^*=\inf\{N: \sum_{n=1}^N  Y_n\geq \rho A \}$.
Therefore,
$$
\Prob_{{\gamma}}\Big(\sup_t \widetilde V^{(1)}_t\geq \rho A\Big) = \Prob_{{\gamma}}\Big(N^*<\infty\Big).
$$
We apply a change of measure to provide an upper bound to the above expression. Let $\widetilde \Prob$ and $\widetilde \ProbQ$ be probability measures under which $Y_n, n=1,2,...$ are i.i.d. random variables and
$$
\widetilde \Prob(Y_n={a})= p \mbox{ and } \widetilde \Prob(Y_n=-{b}) =1-p,
$$
and
$$
\widetilde \ProbQ(Y_n={a})= q \mbox{ and } \widetilde \Prob(Y_n=-{b}) =1-q,
$$
where $p=(e^{{a}}-e^{-{b}})^{-1}(1-e^{-{b}})$ and $q=(e^{{a}}-e^{-{b}})^{-1}e^{{a}}(1-e^{-{b}})$. With a change of measure, we have
\begin{equation}\label{changeofmeasure}
\Prob_{{\gamma}}(N^*<\infty)= E^{\widetilde \ProbQ} \Big[\frac{d\Prob_{N^*}}{d\widetilde \Prob_{N^*}}\frac{d\widetilde \Prob_{N^*}}{d \widetilde \ProbQ_{N^*}} ;N^*<\infty \Big],
\end{equation}
where $\frac{d\Prob_{N^*}}{d\widetilde \Prob_{N^*}}$ and
$
\frac{d\widetilde \Prob_{N^*}}{d \widetilde \ProbQ_{N^*}}
$
denote the likelihood ratios between 
$\Prob_{\gamma}$ and $\widetilde{\Prob}$, and between $\widetilde\Prob$ and $\widetilde \ProbQ$  at the stopping time $N^*$ respectively.
It is easy to check that
$$
\frac{d\widetilde \Prob_{N^*}}{d \widetilde \ProbQ_{N^*}}=\exp{\Big(\sum_{n=1}^{N^*}Y_n\Big)}.
$$
Because $N^*<\infty$ implies $\sum_{n=1}^{N^*} Y_n\geq \rho A$, the probability in \eqref{changeofmeasure} has an upper bound
$$
e^{-A} E^{\widetilde \ProbQ} \Big[\frac{d\Prob_{N^*}}{d\widetilde \Prob_{N^*}};N^*<\infty\Big].
$$
$E^{\widetilde \ProbQ} \Big[\frac{d\Prob_{N^*}}{d\widetilde \Prob_{N^*}};N^*<\infty\Big]$
can be written as the sum
\begin{equation}\label{upperbound}
\underbrace{E^{\widetilde \ProbQ}\Big[ \frac{d\Prob_{N^*}}{d\widetilde \Prob_{N^*}}; N^*\leq\kappa \frac{A}{\bar h}\Big]}_{I_1}+\underbrace{\sum_{k=\kappa+1}^{\infty}E^{\widetilde \ProbQ}\Big[ \frac{d\Prob_{N^*}}{d\widetilde \Prob_{N^*}} ; k\frac{A}{\bar h}\leq N^*\leq(k+1) \frac{A}{{a}}\Big]}_{I_2}.
\end{equation}
It is sufficient to show that $I_1+I_2$
can be bounded by $e^{o(A)}$ as $A\to\infty$ for some constant $\kappa$ that is sufficiently large. We provide upper bounds for $I_1$ and $I_2$ separately. We start with an upper bound for $I_1.$
Notice that under $\Prob_{{\gamma}}$, $Y_n,n=1,2,...$ are i.i.d. random variables and
$$
\Prob_{{\gamma}}(Y_n={a})=\widetilde\alpha_\gamma \mbox{ and } \Prob_{{\gamma}}(Y_n=-{b})=1-\widetilde\alpha_\gamma,
$$
where $\widetilde \alpha_\gamma$ is defined in \eqref{eq:local_error}; then
\begin{equation}\label{ratiobound}
\frac{d\Prob_{n}}{d\widetilde \Prob_{n}}=\lb\frac{\widetilde{\alpha}_\gamma}{p}\rb^{\# \{i: Y_i={a}, \mbox{ and } i\leq n\}}\lb\frac{1-\widetilde{\alpha}_\gamma}{1-p}\rb^{\#\{i: Y_i=-{b}, \mbox{ and } i\leq n \}}\leq e^{o({a})n}.
\end{equation}
The second inequality is due to $\widetilde \alpha_\gamma\leq e^{-(1+o(1)){a}}$ according Proposition~\ref{C_GSPRT_perf}, and $p=e^{-{a}}(1+o(1))$ as ${a},{b}\to\infty$.
Consequently, 
\begin{equation}\label{firstterm}
I_1 \leq e^{\kappa A o({a})/{a}}\leq e^{o(A)}.
\end{equation}
We proceed to an upper bound of $I_2$. According to \eqref{ratiobound}, we have
\begin{equation}\label{secondterm}
I_2\leq \sum_{i=\kappa+1}^{\infty} e^{(k+1)o(A)}\widetilde \ProbQ\Big( \sup_{1\leq n\leq k\frac{A}{{a}}} \sum_{i=1}^{n}Y_i <A \Big)\leq \sum_{k=\kappa+1}^{\infty} e^{(k+1)o(A)}\widetilde \ProbQ\lb \sum_{i=1}^{\lfloor k\frac{A}{{a}}\rfloor}Y_i <A\rb.
\end{equation}
The event $\Big\{\sum_{i=1}^{\lfloor k\frac{A}{{a}}\rfloor}Y_i <A\Big\}$ implies that $\#\{i: Y_i=-{b}, \mbox{ and } i\leq n \}\geq \frac{(k-1)A}{{b}}$. Therefore,
$$
\ProbQ\lb \sum_{i=1}^{\lfloor k\frac{A}{{a}}\rfloor}Y_i <A\rb\leq \widetilde \ProbQ\Big(\#\Big\{i: Y_i=-{b}, \mbox{ and } i\leq \lfloor k\frac{A}{{a}}\rfloor \Big\}\geq \frac{(k-1)A}{{b}}\Big).
$$
Standard result on tail bound for binomial distribution (see, for example, \cite{chernoff1952measure}) yields
\begin{align}\label{binomtail}
&\ProbQ\Big(\#\{i: Y_i=-{b}, \mbox{ and } i\leq \lfloor k\frac{A}{{a}}\rfloor \}\geq \frac{(k-1)A}{{b}}\Big)\nonumber\\&\quad\quad\qquad\qquad\qquad\leq \exp\Big(-2\frac{[(k-1)A/{b}]^2}{\lfloor k\frac{A}{{a}}\rfloor (1-q)}\Big)\leq \exp\Big(-\varepsilon k\frac{A}{{a}} e^{{b}}\Big)\leq e^{-\varepsilon k A},
\end{align}
for some positive constant $\varepsilon$ that is independent of $A$ and ${a}$.
Combining \eqref{secondterm} and \eqref{binomtail}, we have
$$
I_2\leq \sum_{k=\kappa}^{\infty} e^{(k+1)o(A)}e^{-\varepsilon kA}\leq e^{-\varepsilon \kappa A}.
$$
We complete the proof by combining the upper bounds for $I_1$ and $I_2$.
\end{proof}

\bibliographystyle{IEEEtran}
\bibliography{IEEEabrv,references}

\begin{thebibliography}{10}
\providecommand{\url}[1]{#1}
\csname url@samestyle\endcsname
\providecommand{\newblock}{\relax}
\providecommand{\bibinfo}[2]{#2}
\providecommand{\BIBentrySTDinterwordspacing}{\spaceskip=0pt\relax}
\providecommand{\BIBentryALTinterwordstretchfactor}{4}
\providecommand{\BIBentryALTinterwordspacing}{\spaceskip=\fontdimen2\font plus
\BIBentryALTinterwordstretchfactor\fontdimen3\font minus
  \fontdimen4\font\relax}
\providecommand{\BIBforeignlanguage}[2]{{%
\expandafter\ifx\csname l@#1\endcsname\relax
\typeout{** WARNING: IEEEtran.bst: No hyphenation pattern has been}%
\typeout{** loaded for the language `#1'. Using the pattern for}%
\typeout{** the default language instead.}%
\else
\language=\csname l@#1\endcsname
\fi
#2}}
\providecommand{\BIBdecl}{\relax}
\BIBdecl

\bibitem{Poor_DE}
H.~V. Poor, \emph{An Introduction to Signal Detection and Estimation},
  2nd~ed.\hskip 1em plus 0.5em minus 0.4em\relax New York: Springer, 1994.

\bibitem{WaldWolf48}
A.~Wald and J.~Wolfowitz, ``Optimum character of the sequential probability
  ratio test,'' \emph{The Annals of Mathematical Statistics}, vol.~19, no.~3,
  1948.

\bibitem{Lorden76}
G.~Lorden, ``{2-SPRT's} and the modified {Kiefer-Weiss} problem of minimizing
  an expected sample size,'' \emph{The Annals of Statistics}, vol.~4, no.~2,
  pp. 281--291, 1976.

\bibitem{PollakSiegmund75}
M.~Pollak and D.~Siegmund, ``Approximations to the expected sample size of
  certain sequential tests,'' \emph{The Annals of Statistics}, vol.~3, no.~6,
  pp. 1267--1282, 1975.

\bibitem{Lai88}
T.~L. Lai, ``Boundary crossing problems for sample means,'' \emph{Annals of
  Probability}, vol.~28, no.~1, pp. 57--74, 1988.

\bibitem{LaiZhang94}
T.~L. Lai and L.~Zhang, ``A modification of schwarz's sequential likelihood
  ratio tests in multivariate sequential analysis,'' \emph{Sequential
  Analysis}, vol.~13, no.~2, pp. 79--96, 1994.

\bibitem{Lai01}
T.~L. Lai, ``Sequential analysis: Some classical problems and new challenges,''
  \emph{Statistica Sinica}, vol.~11, no.~2, pp. 303--408, Apr. 2001.

\bibitem{Pavlov87}
I.~V. Pavlov, ``A sequential procedure for testing many composite hypotheses,''
  \emph{Theory of Probability and its Applications}, vol.~21, no.~1, pp.
  138--142, 1987.

\bibitem{Pavlov90}
------, ``Sequential procedure of testing composite hypotheses with
  applications to the kiefer-weiss problem,'' vol.~35, no.~2, pp. 280--292,
  1990.

\bibitem{Veeravalli93}
V.~V. Veeravalli, T.~Ba\c{s}ar, and H.~V. Poor, ``Decentralized sequential
  detection with a fusion center performing the sequential test,'' \emph{{IEEE}
  Trans. Inf. Theory}, vol.~39, no.~2, pp. 433--442, Mar. 1993.

\bibitem{Tsitsiklis93}
J.~N. Tsitsiklis, ``Decentralized detection,'' \emph{Advances in Statistical
  Signal Processing}, vol.~2, pp. 297--344, 1993.

\bibitem{Tsitsiklis86}
------, ``On threshold rules in decentralized detection,'' in \emph{Proc. 25th
  Conference on Decision and Control}, Athens, Greece, Dec. 1986, pp. 232--236.

\bibitem{Nguyen06}
X.~Nguyen, M.~J. Wainwright, and M.~I. Jordan, ``On optimal quantization rules
  for sequential decision problems,'' in \emph{Proc. IEEE Int. Symp. Inf.
  Theory}, Seattle, WA, Jul. 9-14 2006.

\bibitem{Mei08}
Y.~Mei, ``Asymptotic optimality theory for decentralized sequential hypothesis
  testing in sensor networks,'' \emph{{IEEE} Trans. Inf. Theory}, vol.~54,
  no.~5, pp. 2072--2089, May 2008.

\bibitem{Wang11}
Y.~Wang and Y.~Mei, ``Asymptotic optimality theory for decentralized sequential
  multihypothesis testing problems,'' \emph{{IEEE} Trans. Inf. Theory},
  vol.~57, no.~10, pp. 7068--7083, Oct. 2011.

\bibitem{Wang13}
------, ``Quantization effect on the log-likelihood ratio and its application
  to decentralized sequential detection,'' \emph{{IEEE} Trans. Signal
  Process.}, vol.~61, no.~6, pp. 1536--1543, Mar. 2013.

\bibitem{Veeravalli94}
V.~V. Veeravalli, T.~Ba\c{s}ar, and H.~V. Poor, ``Decentralized sequential
  detection with sensors performing sequential tests,'' \emph{Mathematics of
  Control, Signals and Systems}, vol.~7, no.~4, pp. 292--305, 1994.

\bibitem{Hussain94}
A.~M. Hussain, ``Multisensor distributed sequential detection,'' \emph{{IEEE}
  Trans. Aerosp. Electron. Syst.}, vol.~30, no.~3, pp. 698--708, Jul. 1994.

\bibitem{Fellouris11}
G.~Fellouris and G.~V. Moustakides, ``Decentralized sequential hypothesis
  testing using asynchronous communication,'' \emph{{IEEE} Trans. Inf. Theory},
  vol.~57, no.~1, pp. 534--548, Jan. 2011.

\bibitem{Yasin12}
Y.~Yilmaz, G.~Moustakides, and X.~Wang, ``Cooperative sequential spectrum
  sensing based on level-triggered sampling,'' \emph{{IEEE} Trans. Signal
  Process.}, vol.~60, no.~9, pp. 4509--4524, Sep. 2012.

\bibitem{Yasin13}
------, ``Channel-aware decentralized detection via level-triggered sampling,''
  \emph{{IEEE} Trans. Signal Process.}, vol.~61, no.~2, pp. 300--315, Jan.
  2013.

\bibitem{Kar12}
S.~Kar, H.~Chen, and P.~K. Varshney, ``Optimal identical binary quantizer
  design for distributed estimation,'' \emph{{IEEE} Trans. Signal Process.},
  vol.~60, no.~7, pp. 3896--3901, Jul. 2012.

\bibitem{Fang13}
J.~Fang, Y.~Liu, H.~Li, and S.~Li, ``One-bit quantizer design for multisensor
  {GLRT} fusion,'' \emph{{IEEE} Signal Process. Lett.}, vol.~20, no.~8, Mar.
  2013.

\bibitem{Ciuonzo13}
D.~Ciuonzo, G.~Papa, G.~Romano, P.~S. Rossi, and P.~Willett, ``One-bit
  decentralized detection with a {Rao} test for multisensor fusion,''
  \emph{{IEEE} Signal Process. Lett.}, vol.~20, no.~9, pp. 861--864, Sep. 2013.

\bibitem{Tartakovsky08}
A.~G. Tartakovsky and A.~S. Polunchenko, ``Quickest changepoint detection in
  distributed multisensor systems under unknown parameters,'' in \emph{Proc.
  11th International Conference on Information Fusion}, Cologne, Germany, 30
  June-3 July 2008.

\bibitem{SLi14_Rep}
S.~Li and X.~Wang, ``Quickest attack detection in multi-agent reputation
  systems,'' \emph{{IEEE} J. Sel. Topics Signal Process.}, vol.~8, no.~4, pp.
  653--666, Aug. 2014.

\bibitem{SLi14_SG}
S.~Li, Y.~Y{\i}lmaz, and X.~Wang, ``Quickest detection of false data injection
  attack in wide-area smart grid,'' \emph{IEEE Trans. Smart Grid}, to appear.
  [Online]. Available: \url{10.1109/TSG.2014.2374577}.

\bibitem{XLi14}
X.~Li, J.~Liu, and Z.~Ying, ``Generalized sequential probability ratio test for
  separate families of hypotheses,'' \emph{Sequential Analysis}, vol.~33,
  no.~4, pp. 539--563, Oct. 2014.

\bibitem{Blum97}
R.~S. Blum, S.~A. Kassam, and H.~V. Poor, ``Distributed detection with multiple
  sensors: {Part II}--advanced topics,'' \emph{Proceedings of IEEE}, vol.~85,
  no.~1, pp. 64--79, Jan. 1997.

\bibitem{Tsitsiklis_TC93}
J.~N. Tsitsiklis, ``Extremal properties of likelihood ratio quantizers,''
  \emph{{IEEE} Trans. Commun.}, vol.~41, no.~4, pp. 550--558, 1993.

\bibitem{Segura10}
J.~Font-Segura and X.~Wang, ``{GLRT-Based Spectrum Sensing for Cognitive Radio
  with Prior Information},'' \emph{{IEEE} Trans. Commun.}, vol.~58, no.~7, pp.
  2137--2146, Jul. 2010.

\bibitem{chernoff1952measure}
H.~Chernoff, ``A measure of asymptotic efficiency for tests of a hypothesis
  based on the sum of observations,'' \emph{The Annals of Mathematical
  Statistics}, pp. 493--507, 1952.

\end{thebibliography}
\end{document}